\documentclass[compsoc,conference,a4paper,10pt]{IEEEtran}
\IEEEoverridecommandlockouts
\def\BibTeX{{\rm B\kern-.05em{\sc i\kern-.025em b}\kern-.08em
    T\kern-.1667em\lower.7ex\hbox{E}\kern-.125emX}}

\usepackage[adversary,sets,asymptotics,landau,lambda,operators,logic,primitives]{cryptocode}
\usepackage[most]{tcolorbox}
\usepackage{fmtcount} % for zero-padded numbers
\usepackage[inline]{enumitem}
\usepackage{xspace}
\usepackage{amsmath} %%% better math %%%
\newcommand{\paragraphNoSkip}[1]{\par\smallskip\noindent\textbf{#1.}\hspace{-0.1em}}
\usepackage[backend=biber,style=numeric,sorting=nyt,maxbibnames=10,maxalphanames=10,minalphanames=10]{biblatex}
\usepackage[T1]{fontenc}
\usepackage[english]{babel}
\usepackage{csquotes}

\usepackage{amsthm}
\newtheorem{definition}{Definition}[section]
\newtheorem{theorem}[definition]{Theorem}

\theoremstyle{remark}
\usepackage{graphicx}
\graphicspath{{./figures/}}
\usepackage{booktabs}
\usepackage{array,multirow}
\usepackage{url}
\PassOptionsToPackage{hyphens}{url}

\usepackage{hyperref}
\PassOptionsToPackage{breaklinks}{hyperref}
\usepackage{xcolor}
\tcbset{
  protocol/.style={
    enhanced,
    boxrule=0.5pt,
    colback=gray!5,
    colframe=black,
    coltitle=black,
    title={#1},
    sharp corners,
    fontupper=\footnotesize,
    before upper={\vspace{1mm}\parindent15pt}, % indent steps
    after title={\par\smallskip},
    attach boxed title to top left={
      yshift=-3mm, xshift=4mm
    },
    boxed title style={
      colback=gray!5,
      colframe=black,
      fonttitle=\fontsize{6}{8}\selectfont,
      boxrule=0.5pt,
    }
  },
  protocoltee/.style={
    enhanced,
    boxrule=1pt,
    colback=gray!30,
    colframe=black,
    coltitle=black,
    title={#1},
    sharp corners,
    fontupper=\footnotesize,
    before upper={\vspace{1mm}\parindent15pt}, % indent steps
    after title={\par\smallskip},
    attach boxed title to top left={
      yshift=-3mm, xshift=4mm
    },
    boxed title style={
      colback=gray!30,
      colframe=black,
      fonttitle=\fontsize{6}{8}\selectfont,
      boxrule=1pt,
    }
  }
}
\newcommand{\step}[1]{\noindent\underline{\textsc{#1}}\par}
\newcommand{\hlineprotocol}{\par\noindent\rule{\linewidth}{0.4pt}\par}
\newlist{codelines}{enumerate}{1}
\setlist[codelines,1]{
  label=\texttt{\padzeroes[2]{\arabic*}:}, % 2-digit padded numbers like 01:, 02:
  leftmargin=2.5em,
  labelwidth=2em,
  labelsep=0.5em,
  itemsep=0pt,
  parsep=0pt,
  font=\ttfamily
}

\usepackage{enumitem}
\usepackage{parskip}
\setenumerate[1]{label={(\arabic*)}} % Global setting
\usepackage{algorithm}
\usepackage{algpseudocode}
\usepackage[capitalise]{cleveref}
\crefname{definition}{Definition}{Definitions}
\crefname{proposition}{Proposition}{Propositions}
\crefname{corollary}{Corollary}{Corollaries}
\crefname{example}{Example}{Examples}
\crefname{remark}{Remark}{Remarks}

\newcommand{\Rep}{\ensuremath{\mathsf{Rep}}}
\newcommand{\Setup}{\mathsf{Setup}}
\newcommand{\KeyGen}{\mathsf{KeyGen}}
\newcommand{\Sign}{\mathsf{Sign}}
\newcommand{\Verify}{\mathsf{Ver}}
\newcommand{\Aggregate}{\mathsf{Agg}}
\newcommand{\AggVerify}{\mathsf{AggVer}}
\newcommand{\Hash}{\mathsf{Hash}}
\newcommand{\Register}{\mathsf{Register}}
\newcommand{\Announce}{\mathsf{Announce}}
\newcommand{\Report}{\mathsf{Report}}
\newcommand{\Attest}{\mathsf{Attest}}
\newcommand{\Migrate}{\mathsf{Migrate}}
\newcommand{\SCRead}{\mathsf{SC.Read}}
\newcommand{\SCWrite}{\mathsf{SC.Write}}
\newcommand{\SCInit}{\mathsf{SC.Init}}
\begin{document}

\title{Persistent BitTorrent Trackers}

\author{
\IEEEauthorblockN{François-Xavier Wicht} \IEEEauthorblockA{University of Bern, IC3} \and % 0009-0005-6090-7901
\IEEEauthorblockN{Zhengwei Tong} \IEEEauthorblockA{Duke University} \and % 0009-0006-9724-0837
\IEEEauthorblockN{Shunfan Zhou} \IEEEauthorblockA{Phala Network} \and % 0009-0006-1255-4045
\IEEEauthorblockN{Hang Yin} \IEEEauthorblockA{Phala Network} \and
\IEEEauthorblockN{Aviv Yaish} \IEEEauthorblockA{Yale University, IC3,\\Complexity Science Hub Vienna} % 0000-0002-7971-2494
}

\maketitle

\begin{abstract}
Private BitTorrent trackers enforce upload-to-download ratios to prevent free-riding, but suffer from three critical weaknesses: reputation cannot move between trackers, centralized servers create single points of failure, and upload statistics are self-reported and unverifiable.
When a tracker shuts down, users lose their contribution history and cannot prove their standing to new communities.

We address these problems by storing reputation in smart contracts and replacing self-reports with cryptographic attestations.
Peers sign receipts for received pieces; the tracker aggregates them via BLS signatures and updates reputation.
If a tracker is unavailable, peers fall back to an authenticated distributed hash table (DHT): stored reputation acts as a public key infrastructure (PKI), preserving access control without the tracker.
Reputation is portable across tracker failures through single-hop migration in factory-deployed contracts.
We also address the privacy implications of publishing public keys and reputations tied to private trackers on a public ledger: we propose ephemeral session keys to prevent linking peer identities, zero-knowledge membership proofs for anonymous DHT participation, and confidential reputation using homomorphic commitments.

We formalize the security requirements, prove four security properties under standard cryptographic assumptions, and evaluate a prototype.
Measurements show that transfer receipts add less than 5\% end-to-end overhead with typical piece sizes. To minimize signing overhead, we adopt a hybrid signature scheme: ECDSA signs individual piece receipts at transfer time for low per-operation latency, while BLS serves as the overarching scheme, enabling compact aggregation of many receipts into a single proof at report time. This design reduces client-side signing cost by an order of magnitude compared to using BLS throughout.
\end{abstract}

\begin{IEEEkeywords}
p2p, file transfer, censorship resistance, persistent reputation management, distributed file exchange.
\end{IEEEkeywords}

\section{Introduction}
\label{sec:intro}
BitTorrent has become one of the most widely deployed peer-to-peer (p2p) protocols, responsible for a significant fraction of global internet traffic.
While public BitTorrent trackers allow unrestricted participation, they suffer from free-riding: users who download content without contributing equivalent uploads~\cite{adar2000free,jun2005incentives}.
Private\footnote{Following BitTorrent's standard terminology, ``private'' does not refer to privacy notions such as anonymity or unlinkability, but rather to permissioned resources (as opposed to publicly accessible ones).} trackers emerged as a community-driven solution, restricting access to users who maintain favorable upload-to-download ratios.
These systems store reputation in centralized databases, creating vibrant but fragile communities vulnerable to single points of failure.
This design hampers quick recovery from failures, as illustrated by a notable incident.
In 2007, European police agencies shut down the prominent OiNK tracker, which was called ``the world's greatest record store'' \cite{westhoff2007trent}, with its founder even named one of online music's most influential people \cite{dolan2007powergeek}.
While alternative trackers were rapidly created, admissions were often invitation-only and ex-OiNK members could not migrate their hard-earned ratios  \cite{fisher2007oinks}. 

The private tracker model exhibits three critical structural weaknesses.
\emph{First}, upload-to-download ratios are not portable across trackers.
Communities operate as isolated silos, preventing users from leveraging their contribution history when joining new trackers or recovering from shutdowns.
\emph{Second}, centralization creates fragility: trackers serve as single points of failure for both reputation storage and peer discovery.
While distributed solutions like DHTs~\cite{maymounkov2002kademlia} can handle peer discovery, they lack authentication and are thus disabled for private tracker torrents~\cite{BortnikovGKKS08}.
\emph{Third}, transfer statistics are self-reported and unverifiable.
Users can inflate ratios through false reports~\cite{brooks2009bittorrhacks}, with only ex post moderation available to detect fraud.

We redesign the private tracker architecture to eliminate these three weaknesses through blockchain-based reputation and cryptographic attestation.

\emph{First}, we persist reputation and make it portable through smart contracts that record user contributions on-chain.
Users can migrate reputation to new trackers, join federated communities, or bootstrap new tracker instances.
When a tracker shuts down, no reputation is lost: the blockchain preserves all historical contributions, enabling seamless migration to successor communities.

\emph{Second}, we eliminate both forms of centralization.
For reputation storage, smart contracts replace centralized databases, thus no single entity controls or can destroy reputation data.
The tracker posts cryptographically authenticated state transitions to the blockchain for accountability.
Should a tracker fail or become compromised, the contract enables rollback to the last consistent state.
For peer discovery, on-chain reputation serves as an authenticated allow-list, letting peers fall back on DHT-based discovery when trackers are down.
This decentralized fallback maintains access control and eliminates the single point of failure risk posed by the tracker in peer discovery.

\emph{Third}, we introduce an attestation protocol where peers cryptographically sign evidence of data transfers.
The tracker aggregates attestations, creating an auditable chain of custody for reported statistics.
Senders cannot inflate contributions without obtaining receivers' signatures, and disputes can be resolved by examining cryptographic receipts rather than relying on ex post moderation.

\emph{Optionally}, when the tracker operator cannot be trusted, the tracker software can run inside a Trusted Execution Environment (TEE).
This adds two guarantees on top of the base protocol: \emph{integrity}, ensuring the operator cannot selectively drop valid attestations to disadvantage certain peers; and \emph{confidentiality}, preventing the operator from observing peer IP addresses and activity patterns.
The TEE is not required when the operator is trusted, which is the common case in small private communities.

The result is a \emph{persistent tracker}: a censorship-resistant protocol for private content distribution that aligns with Web3 principles of decentralization, verifiability, and user sovereignty.
We formalize the security requirements, present a construction with per-piece attestation, and demonstrate how blockchain-based persistent storage combined with cryptographic mechanisms achieves robust guarantees even against powerful adversaries.

In total, we contribute:
\begin{enumerate*}
  \item A formal \emph{Persistent BitTorrent Tracker Scheme} (PBTS) with security requirements, algorithms, and proofs under standard cryptographic assumptions.
  \item A construction adding verifiable per-piece attestation to BitTorrent, replacing self-reported statistics with cryptographically checked transfers.
  \item Optimizations for attestation to reduce signing cost, bandwidth, and verification overhead.
  \item A portable reputation system using factory-based smart contracts where new trackers inherit state from predecessors through single-hop migration.
  \item An authenticated DHT fallback using on-chain reputation as PKI, maintaining access control when trackers are unavailable, and an optional TEE-based mode that extends the base protocol with operator-integrity and peer-privacy guarantees.
  \item A privacy analysis and set of mitigations: ephemeral session keys to prevent IP-to-identity linking in swarms, zero-knowledge membership proofs for anonymous DHT participation, and confidential reputation via homomorphic commitments.
  \item An implementation and evaluation showing that PBTS achieves strong security guarantees with less than 5\% throughput overhead for typical workloads, and that a hybrid ECDSA/BLS signing scheme reduces client-side signing cost by an order of magnitude for high-bandwidth transfers.
\end{enumerate*}

The paper is structured as follows.
\Cref{section:RelatedWork} surveys related work.
\Cref{sec:prelims} provides background on the BitTorrent protocol and establishes the needed notation and cryptographic primitives.
\Cref{sec:Persistent BitTorrent Tracker} formally defines the Persistent BitTorrent Tracker Scheme, presents our construction with per-piece attestation, factory-based smart contracts for portable reputation, and authenticated DHT fallback for tracker-less operation.
\Cref{sec:security} analyzes the security properties of the construction.
\Cref{sec:privacy} analyzes privacy risks and proposes mitigations including ephemeral session keys, ZK membership proofs, and confidential reputation.
\Cref{sec:implementation} describes our implementation and evaluates the prototype.
Finally, we conclude in \Cref{section:Conclusion}.

\section{Related work}
\label{section:RelatedWork}

PBTS sits at the intersection of four research areas. Incentive design in p2p file-sharing motivates the problem; reputation and identity management inform our on-chain design; censorship-resistant distributed systems shape our persistence and migration mechanisms; and trusted execution underpins the optional integrity guarantees.

\paragraphNoSkip{Fairness in file-sharing}
Ensuring fairness is a long-standing challenge in open p2p systems, e.g., preventing ``free-riding'' by users \cite{adar2000free}, that is, users who only download content without uploading at least an equivalent amount to others.
Some have proposed mitigating this issue by requiring micro-payments for downloads, possibly by protocol-specific currencies \cite{golle2001incentives}.
BitTorrent attempts to address this via an optional ``choking'' protocol where a user may temporarily refuse to upload to peers who do not reciprocate \cite{cohen2003incentives}.
The analysis of \citeauthor{wu2007proportional} \cite{wu2007proportional} shows that such tit-for-tat protocols quickly converge to an efficient equilibrium: bandwidth is optimally allocated.

\paragraphNoSkip{Private trackers}
Private trackers emerged as a community-driven solution to free-riding, introducing admission-control and a reputation layer based on upload-to-download ratios to enforce sharing norms \cite{kash2012economics}.
Thus, communities typically only admit new users with a good upload-to-download ratio in other communities, and kick out existing users who do not maintain a good ratio.
The study of \citeauthor{hales2009bittorrent} \cite{hales2009bittorrent} identifies potential ``credit squeezes'' where a lack of upload opportunities can stifle participation in private trackers.
While some propose to improve fairness by applying economic inequality measures (e.g., the Gini coefficient) to file-sharing communities \cite{rahman2010improving}, recent work finds that such measures are inaccurate in pseudonymous settings \cite{yaish2025inequality}.

\paragraphNoSkip{Sybil attacks in BitTorrent}
A notable issue that may arise in BitTorrent communities is that actors can fake their upload-to-download ratio, whether by falsely reporting uploads \cite{brooks2009bittorrhacks}, or by creating multiple identities who upload and download from each other.
The latter manipulation is part of a broader class of so-called \emph{Sybil} attacks that involve the creation of ``fake'' identities \cite{douceur2002sybil}.
While BitTorrent's tit-for-tat is resistant to some manipulations \cite{cheng2023truthfulness}, it is vulnerable to Sybil attacks \cite{levin2008bittorrent}.
\citeauthor{cheng2022study} \cite{cheng2022study} show that the gain that can be obtained by such attacks equals at most three times the amount of data that could be downloaded honestly, with a tight bound of two obtained later by \citeauthor{cheng2024tight} \cite{cheng2024tight}.
Prior work analyzes the economic impact of file-sharing services on a market \cite{minniti2010turning}, and how such services should be priced \cite{li2021optimal}, with later work showing that when incentives are not correctly aligned, Sybil attacks may drain victim resources \cite{halaburda2025platform}.

\paragraphNoSkip{Rethinking BitTorrent's incentives}
Market-based solutions for BitTorrent's vulnerability to Sybil attacks and other manipulations have been explored by prior work.
For example, \citeauthor{levin2008bittorrent} \cite{levin2008bittorrent} provide an elegant analogy: from the perspective of a user, the peers competing for its upload bandwidth are participating in an auction.
The authors find that allocating upload bandwidth proportionally to the incoming bandwidth received from each peer is nearly an equilibrium.
That is, deviating from this strategy is nearly unprofitable, given that all other peers are following the rules.
A different design is offered by \citeauthor{zohar2009adding} \cite{zohar2009adding}, where, by default, peers cannot request specific data blocks, but rather a range from which data blocks are chosen at random.
A user can make specific requests to a peer only in exchange for fully providing blocks asked for by that peer.

\paragraphNoSkip{Reputation management}
A notable line of work proposed systems to persist and manage reputation.
One prominent design for such systems is based on the non-transferable ``soulbound tokens'' of \citeauthor{ohlhaver2022decentralized} \cite{ohlhaver2022decentralized}, which can be used to represent identity, and, by extension, reputation.
The authors emphasize the importance of having a recovery mechanism in place, to assist those who for whatever reason lost control of their tokens (e.g., due to losing the secret key corresponding to the account holding the tokens).
We highlight another crucial recovery notion, for the case where the system itself is compromised.
Alternative reputation management systems similarly lack recovery functionality of this sort, such as the UniRep protocol \cite{unirep2025unirep}.
A general ZKP-based design offering functionality similar to UniRep's is provided by \citeauthor{buterin2022some} \cite{buterin2022some}.
A framework called zk-promises is put forth by \citeauthor{shih2024zk} \cite{shih2024zk}, which re-purposes methods used by privacy-preserving cryptocurrency protocols to endow private reputation systems with moderation capabilities (e.g., to enable blacklisting accounts).

\paragraphNoSkip{Persistent systems}
A main goal of ours is to provide takedown resistance for BitTorrent trackers.
Previous work did not consider this threat model, mostly focusing on network-level interference.
For example, \citeauthor{bocovich2024snowflake} \cite{bocovich2024snowflake} devise an internet-censorship circumvention system which relies on rapidly setting up numerous temporary proxies, which, due to their sheer number, are harder to block.
To reduce the costs associated with launching such proxies, \citeauthor{kon2024spotproxy} \cite{kon2024spotproxy} present a service called SpotProxy which continuously searches for cheap hosting providers, and, if indeed found, deploys new proxy instances and migrates clients to them.
In comparison to our work, SpotProxy relies on a central controller to save client registration details and handle migration.

\paragraphNoSkip{TEEs}
Previous work employed TEEs to design robust systems in other settings and with different objectives than ours.
\citeauthor{DBLP:conf/eurosp/ChengZKHHJJ0S19} \cite{DBLP:conf/eurosp/ChengZKHHJJ0S19} introduce Ekiden, a platform that separates smart contract execution from consensus by running contracts inside SGX enclaves, achieving confidentiality without sacrificing verifiability; \citeauthor{DBLP:journals/popets/LiWWGR22} \cite{DBLP:journals/popets/LiWWGR22} systematize this line of work and compare TEE-backed confidential smart contract designs across key management, attestation, and liveness assumptions.
\citeauthor{zhang2016town} \cite{zhang2016town} extend this approach to authenticated data-feeds for smart contracts, and \citeauthor{maram2021candid} \cite{maram2021candid} use TEE-backed data-feeds for identity: credentials are scraped from websites and transferred on-chain in a trustworthy manner, with MPC used to deduplicate imported credentials and prevent Sybils.
On the systems side, \citeauthor{druschel2001largescale} \cite{druschel2001largescale} use trusted hardware for persistent storage by replicating files across multiple nodes; more recently, \citeauthor{DBLP:conf/sp/ZhuH0WCZWZYZM20} \cite{DBLP:conf/sp/ZhuH0WCZWZYZM20} scale this to rack-level confidential computing using heterogeneous TEEs, and \citeauthor{DBLP:conf/ccs/KuvaiskiiSQXBV24} \cite{DBLP:conf/ccs/KuvaiskiiSQXBV24} demonstrate that unmodified Linux workloads can run inside Intel TDX with low overhead via a lightweight OS kernel.

\section{Preliminaries}
\label{sec:prelims}

We establish notation, review the BitTorrent protocol, and define the cryptographic primitives our construction uses.

\paragraphNoSkip{Notations}
We write $x \sample S$ to denote sampling $x$ uniformly at random from set $S$. A function $\nu: \mathbb{N} \to \mathbb{R}$ is \emph{negligible} if for every polynomial $p(\cdot)$, there exists $N \in \mathbb{N}$ such that for all $n > N$, $\left|\nu(n)\right| < 1/p(n)$.
We write $[n]$ to denote the set $\{1,\ldots,n\}$.

We define reputation as a function $\Rep: \mathbb{R}_{\geq 0} \times \mathbb{R}_{\geq 0} \to \mathbb{R}$ mapping an account's total uploaded and downloaded data to a numerical score. A common instantiation is the sharing ratio $\rho = \frac{\text{uploaded}}{\text{downloaded}}$ (with $\rho = \infty$ when downloaded $= 0$), though trackers may implement alternative reputation functions (e.g., crediting upload contributions more heavily, or incorporating time-weighted statistics). 

\paragraphNoSkip{BitTorrent protocol and tracker architecture}
BitTorrent is a p2p protocol for file distribution. A file is divided into fixed-size pieces, which peers exchange until the full content is reconstructed. Clients advertise which pieces they hold and prioritize peers who reciprocate, enforcing tit-for-tat incentives~\cite{cohen2003incentives}.
Each torrent is described by a \texttt{.torrent} file containing metadata, including the file length, piece hashes, and tracker URLs. Piece hashes guarantee content integrity, while the tracker maps torrent identifiers (infohashes) to active peers. When queried, it returns a random subset of peers for the client to contact.
Once peers discover each other, they exchange handshakes and begin transferring data. Each peer decides locally whom to upload to, based on choking heuristics. 
\emph{Public} trackers permit unrestricted access, while \emph{private} trackers bind user accounts to credentials and enforce upload/download quotas~\cite{5569968}. Account accumulate reputation, typically a download-to-upload ratio. Clients periodically report their statistics to the tracker, which stores them in a centralized database.
These reports are unauthenticated and rely on client honesty. Moderators may intervene to detect fraud, but audits are manual and retrospective.

\paragraphNoSkip{Cryptography}
Our construction relies on three building blocks: aggregatable signatures for efficient batching of peer attestations, smart contracts for persistent reputation storage, and trusted execution environments for authenticated tracker operation.
We formalize each one below.
We also write $\mathsf{Com}(s; r)$ to denote a computationally binding and hiding commitment to value $s$ with randomness $r$, instantiated by any standard commitment scheme.

\begin{definition}[Aggregatable signature scheme]
\label{def:signature}
An aggregatable signature scheme $\Sigma$ consists of five algorithms:
\begin{itemize}
  \item $\KeyGen(1^\lambda) \to (\mathsf{sk}, \mathsf{pk})$: Takes a security parameter $\lambda$ and outputs a key pair consisting of a secret signing key $\mathsf{sk}$ and a public verification key $\mathsf{pk}$.
  \item $\Sign(\mathsf{sk}, m) \to \sigma$: Takes a secret key $\mathsf{sk}$ and a message $m \in \{0,1\}^*$, and outputs a signature $\sigma$.
  \item $\Verify(\mathsf{pk}, m, \sigma) \to \{0,1\}$: Takes a public key $\mathsf{pk}$, a message $m$, and a signature $\sigma$, and outputs $1$ if the signature is valid, or $0$ otherwise.
  \item $\Aggregate(\{(\mathsf{pk}_i, m_i, \sigma_i)\}_{i \in [n]}) \to \sigma_{\mathsf{agg}}$: Takes a set of public keys, messages, and signatures, and outputs an aggregate signature $\sigma_{\mathsf{agg}}$.
  \item $\AggVerify(\{(\mathsf{pk}_i, m_i)\}_{i \in [n]}, \sigma_{\mathsf{agg}}) \to \{0,1\}$: Takes a set of public key-message pairs and an aggregate signature, and outputs $1$ if the aggregate signature is valid for all pairs, or $0$ otherwise.
\end{itemize}
The signature scheme is required to satisfy \emph{correctness}: for all $(\mathsf{sk}, \mathsf{pk}) \leftarrow \KeyGen(1^\lambda)$ and all messages $m$, then $\Verify(\mathsf{pk}, m, \Sign(\mathsf{sk}, m)) = 1$. For aggregate signatures, if each signature $\sigma_i = \Sign(\mathsf{sk}_i, m_i)$ is valid, then: $\AggVerify(\{(\mathsf{pk}_i, m_i)\}_{i \in [n]}, \Aggregate(\{(\mathsf{pk}_i, m_i, \sigma_i)\}_{i \in [n]})) = 1.$

We require the following security properties.

\emph{Strong unforgeability under chosen message attacks} (sUF-CMA): for any PPT adversary $\mathcal{A}$ given $\mathsf{pk}$ and oracle access to $\Sign(\mathsf{sk}, \cdot)$,
\[
  \Pr\bigl[\Verify(\mathsf{pk}, m^*, \sigma^*) = 1 \;\land\; (m^*, \sigma^*) \notin \mathcal{Q}\bigr] \leq \negl,
\]
where $\mathcal{Q}$ is the set of message-signature pairs returned by the oracle. This is strictly stronger than EUF-CMA: the adversary cannot produce a new valid signature even on a previously queried message.

\emph{Aggregate unforgeability}: for any PPT adversary $\mathcal{A}$ with access to signing oracles $\{\Sign(\mathsf{sk}_i, \cdot)\}_{i \in [n]}$,
\begin{align*}
  \Pr\bigl[&\AggVerify(\{(\mathsf{pk}_i, m_i)\}_{i\in[n]},\, \sigma_{\mathsf{agg}}) = 1 \\
           &\;\land\; \exists\, i : m_i \notin \mathcal{Q}_i\bigr] \leq \negl,
\end{align*}
where $\mathcal{Q}_i$ is the message set queried to oracle $i$.
No adversary can produce a valid aggregate with a message-key pair for which it did not obtain a valid signature.

Standard instantiations include BLS signatures~\cite{BonehLS01} over pairing-friendly elliptic curves (e.g., BLS12-381~\cite{bls12-381}), which support efficient signature aggregation~\cite{BonehDN18,bls-aggregation}. 
\end{definition}

Smart contracts~\cite{buterin2013ethereum} serve as the persistent storage layer for reputation data, replacing centralized databases with blockchain-based state that survives tracker failures.
We model smart contracts as programs with authenticated write access and public read access.

\begin{definition}[Smart contract]
\label{def:smart-contract}
A smart contract is a blockchain-deployed deterministic program that maintains persistent state and is invoked by transactions. Our construction requires the following contract operations:

\begin{itemize}
  \item $\SCInit(\mathsf{params}) \to \mathsf{addr}$: Deploys a new contract with initialization parameters $\mathsf{params}$ and returns its on-chain address $\mathsf{addr}$.
  \item $\SCRead(\mathsf{addr}, \mathsf{key}) \to \mathsf{value}$: Reads the value associated with key $\mathsf{key}$ from the contract at address $\mathsf{addr}$. This operation is publicly accessible and does not modify the contract's state.
  \item $\SCWrite(\mathsf{addr}, \mathsf{key}, \mathsf{value}, \mathsf{auth}) \to \{\mathsf{success}, \bot\}$: Writes $\mathsf{value}$ to the contract at address $\mathsf{addr}$ under key $\mathsf{key}$, authenticated by $\mathsf{auth}$. The contract enforces access control: only authorized entities (e.g., the tracker's TEE) can modify state. Returns $\mathsf{success}$ if the write succeeds, or $\bot$ if authorization fails.
\end{itemize}

Smart contracts guarantee \emph{integrity}: state transitions are validated by consensus among blockchain nodes, ensuring that unauthorized modifications are rejected. They also provide \emph{persistence}: once written, data remains immutable and accessible as long as the blockchain operates. In our construction, we use a factory pattern where a single factory contract deploys multiple reputation contracts, each maintaining user statistics for a tracker instance.

More generally, any \emph{bulletin board} primitive that satisfies the following three properties can instantiate our construction in place of a smart contract: \emph{(i) liveness}, meaning honest writes are eventually committed and readable; \emph{(ii) write integrity}, meaning only authorized parties can append or modify entries and committed entries cannot be altered retroactively; and \emph{(iii) censorship resistance}, meaning no single party can prevent a valid write from being recorded. Smart contracts on a sufficiently decentralized blockchain are a natural instantiation, but other designs (e.g., a permissioned ledger operated by a federation of trackers) may satisfy these properties with different trust and cost trade-offs.
\end{definition}

When the tracker operator is trusted (e.g., a well-known community administrator), no additional hardware is required beyond a standard server.
If the operator is not trusted, a TEE provides two concrete benefits: \emph{integrity} (the operator cannot selectively discard valid attestations to disadvantage specific peers) and \emph{confidentiality} (peer IP addresses and activity patterns are hidden from the operator).
These guarantees come at the cost of a hardware trust assumption on the TEE manufacturer.

\begin{definition}[Trusted execution environment]
\label{def:tee}
A \emph{Trusted Execution Environment} (TEE) is a secure area within a processor that provides isolated execution for sensitive code and data. TEEs guarantee confidentiality (data is inaccessible to the host OS), integrity (code cannot be tampered with), and attestation (cryptographic proofs that specific code runs in a genuine TEE).

\end{definition}
Modern VM-level TEE solutions such as Intel TDX~\cite{cheng2024intel} and AMD SEV~\cite{amd2020secure} let unmodified applications run in isolated virtual machines. The attestation mechanism lets third parties cryptographically verify any system component through measurement registers that capture build-time and runtime measurements of the executed code. 

While TEEs provide strong confidentiality and integrity guarantees, they do not guarantee liveness or availability. The host system retains control over resource scheduling, TEE initialization, and system call execution.

Throughout the work, algorithm boxes that optionally execute in a TEE when the operator is untrusted are distinguished by a darker background and a bold border, whereas plain boxes use a lighter background and a thin border.
Code running in a TEE implicitly outputs a certificate of correct execution (attestation).
In trusted-operator deployments these algorithms run as ordinary server code.

\paragraphNoSkip{Threat model}
\label{subsec:threat-model}
The blockchain is trusted for smart contract execution. Peers and DHT nodes are \emph{untrusted}.
We consider two trust settings for the tracker operator.
\emph{Trusted operator} (base model): the operator runs the tracker software honestly; the system delivers portability and verifiable receipts but provides no protection against a dishonest operator.
\emph{Untrusted operator} (TEE-enhanced model): the operator may attempt to manipulate reputation or observe peer data; a TEE (Intel TDX/AMD SEV), trusted for isolation and attestation, prevents tampering with execution and reading of sensitive data. The operator can still jeopardize liveness by terminating instances or denying resources.
The private tracker employs a Sybil-resistant registration mechanism, such as interviews or accountable sponsorship where sponsors are held responsible for invitees' behavior (common in private trackers like RED~\cite{redinterview}), to limit initial account creation.
While peers cannot forge receipts cryptographically, colluding peers could exchange real data to generate legitimate receipts.
However, this has no net benefit: downloaders always record a reputation loss, and accountable sponsorship makes creating accounts costly since sponsors risk penalties for invitees' misconduct.

\section{Persistent BitTorrent tracker system}
\label{sec:Persistent BitTorrent Tracker}

We now introduce our formal specification and construction for the Persistent BitTorrent Tracker System (PBTS).

\subsection{Formal specification}

PBTS extends the traditional tracker interface with multiple algorithms. $\Setup$ and $\KeyGen$ initialize the system and generate user keys. $\Register$ creates accounts authenticated by signatures. $\Announce$ provides peer discovery with reputation-based access control. $\Report$ submits verified transfer statistics backed by aggregated cryptographic receipts. $\Attest$ and $\Verify$ implement p2p attestation, where receivers sign acknowledgments of transfers. $\Migrate$ enables reputation portability by creating new tracker instances that inherit state from predecessors.

\begin{definition}[Persistent BitTorrent tracker scheme]
\label{def:pbts}
A \emph{Persistent BitTorrent Tracker Scheme} (PBTS) is a tuple of eight algorithms
% $\sloppy \mathsf{PBTS} = (\Setup, \KeyGen, \Register, \Announce, \Report, \Attest, \Verify, \Migrate)$,
% with algorithms 
defined as follows:
\begin{itemize}
  \item $\Setup(1^\lambda, \mathsf{MinRep}, \mathsf{InitCredit}, W, \Delta) \to \mathsf{pp}$:
    Takes security parameter $\lambda$, minimum reputation threshold $\mathsf{MinRep}$, initial upload credit $\mathsf{InitCredit}$ for new users, epoch width $W$, and epoch acceptance window $\Delta$, and outputs public parameters $\mathsf{pp}$ including tracker instance ID $\mathsf{iid}$. The public parameters $\mathsf{pp}$ are implicit input to the remaining algorithms.
  \item $\KeyGen() \to (\mathsf{sk}, \mathsf{pk})$:
    Generates user key pair.
  \item $\Register(\mathsf{uid}, \mathsf{pk}, \sigma, \mathsf{params}) \to \{0, 1\}$:
    Registers user with user ID $\mathsf{uid}$, public key $\mathsf{pk}$, and signature $\sigma$ over registration message including instance ID and user ID, with optional parameters $\mathsf{params}$. Returns $1$ if registration succeeds, $0$ otherwise.
  \item $\Announce(\mathsf{uid}, \mathsf{pk}, \sigma, \mathsf{tid}, \mathsf{event}) \to \mathcal{P}$:
    Announces torrent $\mathsf{tid}$ with user ID $\mathsf{uid}$, public key $\mathsf{pk}$, signature $\sigma$ for authentication, and event type $\mathsf{event} \in \{\mathtt{started}, \mathtt{stopped}, \mathtt{completed}, \mathtt{none}\}$, returning peer list $\mathcal{P}$.
  \item $\Report(\mathsf{uid}, \mathsf{pk}, \{\mathsf{pk}_j\}_{j \in \mathcal{J}}, \mathcal{T}, \{t_j\}_{j \in \mathcal{J}}, \sigma_{\mathsf{agg}}, \Delta_{\mathsf{up}}) \to \{0, 1\}$:
    Reports upload statistics with user ID $\mathsf{uid}$, user's public key $\mathsf{pk}$, set of peer public keys $\{\mathsf{pk}_j\}_{j \in \mathcal{J}}$ who provided receipts, torrent metadata $\mathcal{T}$, timestamps $\{t_j\}_{j \in \mathcal{J}}$ for each receipt, aggregated signature $\sigma_{\mathsf{agg}}$, and upload delta $\Delta_{\mathsf{up}}$. The tracker reconstructs receipts, verifies the aggregate signature, credits the reporter's upload, and credits each receipt signer's download counter directly from the receipts. Returns $1$ if accepted, $0$ otherwise.
  \item $\Attest(\mathsf{sk}_{\mathsf{receiver}}, \mathsf{pk}_{\mathsf{sender}}, p_i, \mathcal{T}, t_{\text{epoch}}) \to \sigma_{\mathsf{receipt}}$:
    Generates cryptographic receipt for piece transfer, where receiving peer signs acknowledgment for sending peer. Takes receiver's secret key, sender's public key, piece $p_i$, torrent metadata $\mathcal{T} = (h_{\mathcal{T}}, [h_1, \ldots, h_n])$ containing infohash and piece hashes, and epoch timestamp $t_{\text{epoch}}$. Returns receipt signature binding infohash, sender's public key, piece hash, piece index, and epoch.
  \item $\Verify(\mathsf{pk}_{\mathsf{receiver}}, \mathsf{pk}_{\mathsf{sender}}, p_i, \mathcal{T}, t_{\text{epoch}}, \sigma_{\mathsf{receipt}}) \to \{0, 1\}$:
    Verifies cryptographic receipt by checking receiver's signature and piece integrity against $\mathcal{T}$. Returns $1$ if valid and $0$ otherwise. Valid receipts prove transfers from sender to receiver. Peers use this algorithm with piece data $p_i$ to verify receipts locally, while the tracker only needs the piece hash $h_i$ and index $i$ from $\mathcal{T}$ for signature verification.
  \item $\Migrate(\mathsf{addr}_{\mathsf{old}}, \pi) \to \mathsf{addr}_{\mathsf{new}}$:
    Migrates reputation from old contract address $\mathsf{addr}_{\mathsf{old}}$ using migration proof $\pi$, returning new contract address $\mathsf{addr}_{\mathsf{new}}$.
\end{itemize}
\end{definition}

\subsection{Construction}
\label{subsec:construction}

We now instantiate the PBTS scheme using BLS signatures for receipt aggregation and smart contracts for reputation storage. In the untrusted-operator variant, tracker algorithms additionally run inside a TEE.
\Cref{fig:arch} shows the system architecture.

\begin{figure}[ht!]
  \centering
  \includegraphics[width=\linewidth]{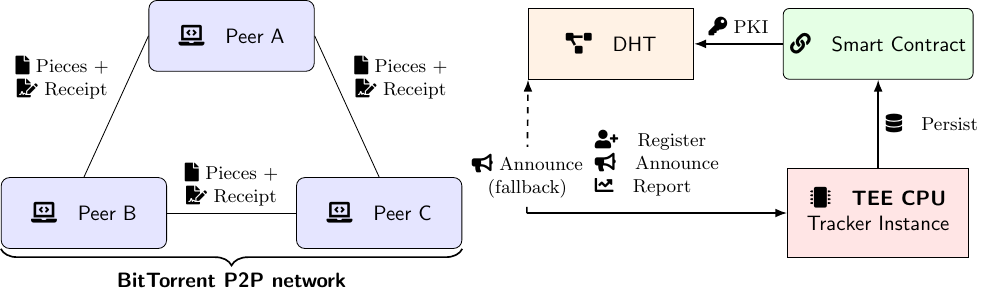}
  \caption{Architecture of the censorship-resistant tracker system. Peers (A, B, C) form a distributed swarm and exchange file pieces directly using the standard BitTorrent protocol. During file transfer, sending and receiving peers exchange cryptographic receipts. Peers register by proving public key ownership, announce torrents to retrieve peer lists with reputation-based access control, and report upload/download statistics with aggregated receipts. The tracker periodically writes reputation data to a smart contract on the blockchain, providing persistent and verifiable reputation storage. In the untrusted-operator variant (shaded), tracker operations execute inside a TEE: this prevents operator manipulation of reputation and protects peer IP addresses from the operator, at the cost of a hardware trust assumption.}
  \label{fig:arch}
\end{figure}

\paragraphNoSkip{Initialization}
Trackers are initialized by running $\Setup$ to generate unique instance identifiers and system parameters, while users generate key pairs via $\KeyGen$ for authentication and signing, as formalized in \cref{fig:pbts-setup}.

\begin{figure}[ht!]
  \begin{tcolorbox}[protocoltee={Setup}]

    \step{$\Setup(1^\lambda, \mathsf{MinRep}, \mathsf{InitCredit}, W, \Delta)$:}
    \begin{codelines}
      \item $\mathsf{iid} \sample \{0,1\}^\lambda$
      \item $\mathsf{pp} \gets (\lambda, \mathsf{iid}, \mathsf{MinRep}, \mathsf{InitCredit}, W, \Delta)$
      \item \textbf{return } $\mathsf{pp}$
    \end{codelines}

\end{tcolorbox}
  \begin{tcolorbox}[protocol={Key generation}]
    \step{$\KeyGen()$:}
    \begin{codelines}
      \item $(\mathsf{sk}, \mathsf{pk}) \gets \Sigma.\KeyGen(1^\lambda)$
      \item \textbf{return } $(\mathsf{sk}, \mathsf{pk})$
    \end{codelines}

  \end{tcolorbox}
  \caption{PBTS initialization. $\Setup$ generates tracker instance ID and system parameters ($\mathsf{MinRep}$, $\mathsf{InitCredit}$, $W$, $\Delta$). $\KeyGen$ generates key pairs for an aggregatable signature scheme.}
  \label{fig:pbts-setup}
\end{figure}

\paragraphNoSkip{User registration}
Users register by proving ownership of their public key through a signature over a registration message containing the atom $\mathtt{register}$, their user ID, and the tracker instance ID. The tracker verifies the signature and checks that the user ID is not already registered by reading from the smart contract. If the user ID already exists, registration is rejected to prevent duplicate accounts. Otherwise, the tracker writes the user's ID, public key, and initial reputation counters to the contract. New users receive an initial upload credit $\mathsf{InitCredit}$ to bootstrap participation, with zero downloads. Registration succeeds only if the contract write operation succeeds. Subsequent operations authenticate users via signatures over operation-specific messages using the registered public key. Cryptographic attestation ensures that only legitimate tracker instances can modify reputation data. The registration algorithm is specified in \cref{fig:pbts-register}.

\begin{figure}[ht!]
  \begin{tcolorbox}[protocoltee={User registration}]

    \step{$\Register(\mathsf{uid}, \mathsf{pk}, \sigma, \mathsf{params})$:}
    \begin{codelines}
      \item $m \gets (\mathtt{register} \parallel \mathsf{iid} \parallel \mathsf{uid})$
      \item \textbf{if } $\Sigma.\Verify(\mathsf{pk}, m, \sigma) = 0$ \textbf{ then return } $0$
      \item $\mathsf{existing} \gets \SCRead(\mathsf{addr}, \mathsf{uid})$
      \item \textbf{if } $\mathsf{existing} \neq \bot$ \textbf{ then return } $0$
      \item $\mathsf{result} \gets \SCWrite(\mathsf{addr}, \mathsf{pk}, (\mathsf{uid}, \mathsf{InitCredit}, 0),\mathsf{auth}_{\mathsf{TEE}})$

      \item \textbf{if } $\mathsf{result} = \bot$ \textbf{ then return } $0$
      \item \textbf{return } $1$
    \end{codelines}

  \end{tcolorbox}
  \caption{User registration with signature verification and on-chain state initialization. The tracker verifies ownership of the public key and writes user credentials to the smart contract with initial reputation counters.}
  \label{fig:pbts-register}
\end{figure}

\paragraphNoSkip{Torrent announcement and peer discovery}
When a peer wishes to participate in a torrent swarm, it announces to the tracker using the $\Announce$ algorithm (\cref{fig:pbts-announce}). The tracker maintains internal state $\mathcal{S}_{\mathsf{tid}}$ for each torrent, storing the set of active peers. The tracker first verifies the peer's signature, then reads the peer's upload and download statistics from the smart contract using their user ID and computes their reputation score. If the peer is starting a new download and their reputation falls below the minimum threshold, access is denied. Otherwise, the tracker updates its internal swarm state: removing the peer if they are stopping, or adding their IP address and port if they are joining or continuing. The tracker then samples a random subset of active peers uniformly at random from the swarm and returns this list to the announcing peer.

\begin{figure}[ht!]
  \begin{tcolorbox}[protocoltee={Torrent announcement}]

    \step{$\Announce(\mathsf{uid}, \mathsf{pk}, \sigma, \mathsf{tid}, \mathsf{event})$:}
    \begin{codelines}
      \item $m \gets (\mathtt{announce} \parallel \mathsf{uid} \parallel \mathsf{tid} \parallel \mathsf{event})$
      \item \textbf{if } $\Sigma.\Verify(\mathsf{pk}, m, \sigma) = 0$ \textbf{ then return } $\emptyset$
      \item $(\mathsf{up}, \mathsf{down}) \gets \SCRead(\mathsf{addr}, \mathsf{uid})$
      \item $r \gets \Rep(\mathsf{up}, \mathsf{down})$
      \item \textbf{if } $\mathsf{event} = \mathtt{started} \land r < \mathsf{MinRep}$ \textbf{ then return } $\emptyset$
      \item \textbf{if } $\mathsf{event} = \mathtt{stopped}$ \textbf{ then } $\mathcal{S}_{\mathsf{tid}} \gets \mathcal{S}_{\mathsf{tid}} \setminus \{(\mathsf{pk}, \cdot, \cdot)\}$
      \item \textbf{else } $\mathsf{ip}, \mathsf{port} \gets$ extract from request;

      $\mathcal{S}_{\mathsf{tid}} \gets \mathcal{S}_{\mathsf{tid}} \cup \{(\mathsf{pk}, \mathsf{ip}, \mathsf{port})\}$
      \item $\mathcal{P} \sample \mathcal{S}_{\mathsf{tid}} \setminus \{(\mathsf{pk}, \cdot, \cdot)\}$
      \item \textbf{return } $\mathcal{P}$
    \end{codelines}

  \end{tcolorbox}
  \caption{Torrent announcement with reputation-based access control. The tracker verifies the signature, computes reputation via $\Rep$ from on-chain statistics, enforces $\mathsf{MinRep}$ for new downloads, updates internal swarm state $\mathcal{S}_{\mathsf{tid}}$, and returns a random sample of peers.}
  \label{fig:pbts-announce}
\end{figure}

\paragraphNoSkip{P2P attestation}
Traditional trackers accept self-reported statistics. We replace this with cryptographic receipts where downloaders sign acknowledgments for received pieces.
Torrent metadata $\mathcal{T} = (h_{\mathcal{T}}, [h_1, h_2, \ldots, h_n])$ contains the infohash $h_{\mathcal{T}}$ and piece hashes $[h_1, \ldots, h_n]$. When peer $A$ uploads piece $p_i$ to $B$, peer $B$ verifies piece integrity ($\Hash(p_i) = h_i$) then generates a cryptographic receipt via $\Attest$ (\cref{fig:pbts-attest-timestamp}), signing a message that binds the infohash $h_{\mathcal{T}}$, sender's public key $\mathsf{pk}_A$, piece hash $h_i$, piece index $i$, and epoch timestamp. Peer $B$ returns this signed receipt to $A$.
Uploaders collect receipts from downloaders as proof of contributions. Later, uploaders report accumulated receipts to the tracker via $\Report$. The tracker verifies the aggregate signature from all downloaders, then credits the uploader's upload counter by the total piece size and each downloader's download counter. The contract serves as PKI: peers retrieve each other's public keys from on-chain registration records to verify receipt signatures.

\begin{figure}[ht!]
  \begin{tcolorbox}[protocol={Piece transfer attestation}]

    \step{$\Attest(\mathsf{sk}_{\mathsf{receiver}}, \mathsf{pk}_{\mathsf{sender}}, p_i, \mathcal{T}, t_{\text{epoch}})$:}
    \begin{codelines}
      \item Parse $\mathcal{T} = (h_{\mathcal{T}}, [h_1, \ldots, h_n])$
      \item \textbf{if } $i \notin [1, n]$ \textbf{ then return } $\bot$
      \item \textbf{if } $\Hash(p_i) \neq h_i$ \textbf{ then return } $\bot$
      \item $m \gets (h_{\mathcal{T}} \parallel \mathsf{pk}_{\mathsf{sender}} \parallel h_i \parallel i \parallel t_{\text{epoch}})$
      \item $\sigma_{\mathsf{receipt}} \gets \Sigma.\Sign(\mathsf{sk}_{\mathsf{receiver}}, m)$
      \item \textbf{return } $\sigma_{\mathsf{receipt}}$
    \end{codelines}

    \hlineprotocol \vspace{.1cm}

    \step{$\Verify(\mathsf{pk}_{\mathsf{receiver}}, \mathsf{pk}_{\mathsf{sender}}, p_i, \mathcal{T}, t_{\text{epoch}}, \sigma_{\mathsf{receipt}})$:}
    \begin{codelines}
      \item Parse $\mathcal{T} = (h_{\mathcal{T}}, [h_1, \ldots, h_n])$
      \item \textbf{if } $i \notin [1, n]$ \textbf{ then return } $0$
      \item \textbf{if } $\Hash(p_i) \neq h_i$ \textbf{ then return } $0$
      \item $m \gets (h_{\mathcal{T}} \parallel \mathsf{pk}_{\mathsf{sender}} \parallel h_i \parallel i \parallel t_{\text{epoch}})$
      \item \textbf{return } $\Sigma.\Verify(\mathsf{pk}_{\mathsf{receiver}}, m, \sigma_{\mathsf{receipt}})$
    \end{codelines}

  \end{tcolorbox}
  \caption{Piece transfer attestation with epoch-based double-spend resistance. $\Attest$ verifies piece integrity and generates a cryptographic receipt binding torrent infohash, sender public key, piece hash, piece index, and epoch. $\Verify$ checks signature and piece integrity. Epoch timestamps prevent receipt reuse.}
  \label{fig:pbts-attest-timestamp}
\end{figure}

\paragraphNoSkip{Statistics reporting}
Uploaders periodically call $\Report$ (\cref{fig:pbts-report-timestamp}) to claim credit for completed transfers. The report contains the receipts collected from downloaders and the claimed upload delta $\Delta_{\mathsf{up}}$. The tracker verifies the aggregate signature over all receipts, increments the reporter's upload counter by $\Delta_{\mathsf{up}}$, and increments each downloader's download counter by $\mathsf{piece\_size}$ per receipt. Receipt timestamps are checked against the current epoch: time is divided into windows of width $W$ seconds, and the tracker accepts receipts from the most recent $\Delta$ epochs, rejecting older ones to prevent replay.

\begin{figure}[ht!]
  \begin{tcolorbox}[protocoltee={Statistics reporting}]

      \step{$\Report(\mathsf{uid}, \mathsf{pk}, \{\mathsf{pk}_j\}_{j \in \mathcal{J}}, \mathcal{T},\{t_j\}_{j \in \mathcal{J}}, \sigma_{\mathsf{a}}, \Delta_{\mathsf{up}})$:}
    \begin{codelines}
      \item Parse $\mathcal{T} = (h_{\mathcal{T}}, [h_1, \ldots, h_n])$
      \item $t_{\text{now}} \gets$ current epoch
      \item \textbf{for each } $j \in \mathcal{J}$:
      \item \quad $t_{\text{epoch},j} \gets \lfloor t_j / W \rfloor$
      \item \quad \textbf{if } $t_{\text{epoch},j} \notin [t_{\text{now}} - \Delta, t_{\text{now}}]$ \textbf{ then return } $0$
      \item \quad $\mathsf{rid}_j \gets (h_{\mathcal{T}}, \mathsf{pk}, \mathsf{pk}_j, h_j, j, t_{\text{epoch},j})$
      \item \quad \textbf{if } $\mathsf{rid}_j \in \mathcal{R}_{\text{recent}}$ \textbf{ then return } $0$
      \item \quad $m_j \gets (h_{\mathcal{T}} \parallel \mathsf{pk} \parallel h_j \parallel j \parallel t_{\text{epoch},j})$
      \item \textbf{if} $\Sigma.\AggVerify(\{(\mathsf{pk}_j, m_j)\}_{j \in \mathcal{J}}, \sigma_{\mathsf{a}}) = 0$ \textbf{then return } $0$
      \item $(\mathsf{up}, \mathsf{down}) \gets \SCRead(\mathsf{addr}, \mathsf{uid})$
      \item $\mathsf{up}' \gets \mathsf{up} + \Delta_{\mathsf{up}}$
      \item $\SCWrite(\mathsf{addr}, \mathsf{uid}, (\mathsf{up}', \mathsf{down}), \mathsf{auth}_{\mathsf{TEE}})$
      \item \textbf{for each } $j \in \mathcal{J}$:
      \item \quad $\mathcal{R}_{\text{recent}} \gets \mathcal{R}_{\text{recent}} \cup \{\mathsf{rid}_j\}$
      \item \quad Retrieve downloader's user ID $\mathsf{uid}_j$ for $\mathsf{pk}_j$
      \item \quad $(\mathsf{up}_j, \mathsf{down}_j) \gets \SCRead(\mathsf{addr}, \mathsf{uid}_j)$
      \item \quad $\mathsf{down}_j' \gets \mathsf{down}_j + \mathsf{piece\_size}$
      \item \quad $\SCWrite(\mathsf{addr}, \mathsf{uid}_j, (\mathsf{up}_j, \mathsf{down}_j'), \mathsf{auth}_{\mathsf{TEE}})$
      \item \textbf{return } $1$
    \end{codelines}

  \end{tcolorbox}
  \caption{Statistics reporting with epoch-based double-spend resistance. Receipts expiry and deduplication are checked before batch verification via $\AggVerify$. IDs are stored in $\mathcal{R}_{\text{recent}}$ with periodic garbage collection.}
  \label{fig:pbts-report-timestamp}
\end{figure}

\paragraphNoSkip{Reputation migration}
When a tracker becomes unavailable, reputation migrates to a new instance via $\Migrate$ (\cref{fig:pbts-migrate}).
The new tracker generates a fresh instance ID and provides a cryptographic attestation of authenticity.
After verification, migration creates a new smart contract that references the predecessor contract, so any observer can trace reputation history back to the original deployment.
Reputation data remains immutable in the old contract and becomes accessible through the new contract's referrer link. Single-level migration prevents complex multi-hop inheritance while enabling tracker continuity.
The end-to-end initialization and migration message flow is detailed in \Cref{fig:workflow}.

\begin{figure}[ht!]
  \begin{tcolorbox}[protocoltee={Reputation migration}]

    \step{$\Migrate(\mathsf{addr}_{\mathsf{old}}, \pi)$:}
    \begin{codelines}
      \item Parse $\pi = (\mathsf{iid}_{\mathsf{new}}, \mathsf{auth}_{\mathsf{TEE}})$
      \item Verify TEE attestation $\mathsf{auth}_{\mathsf{TEE}}$ for new tracker instance
      \item \textbf{if } attestation is invalid \textbf{ then return } $\bot$
      \item $\mathsf{params} \gets (\mathsf{iid}_{\mathsf{new}}, \mathsf{addr}_{\mathsf{old}}, \mathsf{pk}_{\mathsf{tracker}}, \mathsf{auth}_{\mathsf{TEE}})$
      \item $\mathsf{addr}_{\mathsf{new}} \gets \SCInit(\mathsf{params})$
      \item \textbf{return } $\mathsf{addr}_{\mathsf{new}}$
    \end{codelines}

  \end{tcolorbox}
  \caption{Tracker migration with TEE attestation verification. Deploys a new smart contract referencing the old contract as predecessor, establishing single-hop inheritance that preserves reputation continuity.}
  \label{fig:pbts-migrate}
\end{figure}

\paragraphNoSkip{Reputation contract layer}
The \emph{RepFactory} contract uses a factory pattern to deploy per-tracker reputation contracts. Each contract records upload and download counters keyed by user ID, grants exclusive write access to the deploying tracker instance, and optionally references a predecessor contract as a \emph{referrer} for single-hop reputation inheritance on migration. On-chain storage is pseudonymous: only public keys and counters are recorded, with no direct link to real-world identities. We note, however, that this anonymity is shallow and can be broken through traffic analysis or public-key correlation; we discuss privacy risks and mitigations in \cref{sec:privacy}.

\begin{figure}[ht!]
  \begin{tcolorbox}[protocol={Reputation contract}]

    \step{State:}
    \begin{codelines}
      \item $\mathsf{owner} : \mathsf{pk}$ \quad \textit{// tracker's on-chain key}
      \item $\mathsf{referrer} : \mathsf{addr} \cup \{\bot\}$ \quad \textit{// predecessor contract}
      \item $\mathsf{data} : \mathsf{uid} \to (\mathsf{up}, \mathsf{down})$
    \end{codelines}

    \hlineprotocol \vspace{.1cm}

    \step{$\SCInit(\mathsf{iid}, \mathsf{ref}, \mathsf{pk}_{\mathsf{tracker}}, \mathsf{auth})$:}
    \begin{codelines}
      \item \textbf{if} $\Sigma.\Verify(\mathsf{pk}_{\mathsf{tracker}}, \mathsf{iid}, \mathsf{auth}) = 0$ \textbf{then return} $\bot$
      \item $\mathsf{owner} \gets \mathsf{pk}_{\mathsf{tracker}}$
      \item $\mathsf{referrer} \gets \mathsf{ref}$
      \item \textbf{return} newly deployed contract address $\mathsf{addr}$
    \end{codelines}

    \hlineprotocol \vspace{.1cm}

    \step{$\SCRead(\mathsf{uid})$:}
    \begin{codelines}
      \item \textbf{if} $\mathsf{data}[\mathsf{uid}] \neq \bot$ \textbf{then return} $\mathsf{data}[\mathsf{uid}]$
      \item \textbf{if} $\mathsf{referrer} \neq \bot$ \textbf{then return} $\SCRead_{\mathsf{referrer}}(\mathsf{uid})$
      \item \textbf{return} $\bot$
    \end{codelines}

    \hlineprotocol \vspace{.1cm}

    \step{$\SCWrite(\mathsf{uid}, \mathsf{value}, \mathsf{auth})$:}
    \begin{codelines}
      \item \textbf{if} $\Sigma.\Verify(\mathsf{owner}, \mathsf{uid} \parallel \mathsf{value}, \mathsf{auth}) = 0$ \textbf{then return} $\bot$
      \item $\mathsf{data}[\mathsf{uid}] \gets \mathsf{value}$
      \item \textbf{return} $\mathsf{success}$
    \end{codelines}

  \end{tcolorbox}
  \caption{Reputation contract deployed by the RepFactory. $\SCInit$ binds the contract to a tracker key and an optional predecessor. $\SCRead$ performs a single-hop fallback to the predecessor if the user is not found locally. $\SCWrite$ enforces that only the owning tracker can update state.}
  \label{fig:rep-contract-layer}
\end{figure}

\subsection{Tracker-less peer discovery and file exchange}
\label{subsec:dht-fallback}

Private trackers typically disable DHT-based peer discovery because standard DHT protocols lack authentication and access control. Any peer can freely join a swarm, undermining reputation-based admission and access control enforcement. Consequently, private communities mark torrents as \emph{private}, forcing all peer discovery to occur exclusively through a trusted tracker.

We extend Kademlia~\cite{maymounkov2002kademlia} with authentication to provide a DHT fallback preserving access control when the tracker is unavailable. The smart contract serves as PKI: registered users have on-chain public keys and reputation records. Peers authenticate DHT announcements using these credentials, admitting only users with sufficient reputation.

Kademlia is a distributed hash table that maps keys to values without centralized coordination. Each node has a 160-bit identifier, and data are stored on the nodes whose identifiers are closest to a given key under XOR distance $d(x,y) = x \oplus y$. To join, a peer contacts bootstrap nodes from list $\mathcal{B}$ to learn about other participants and builds a routing table of known nodes. To announce availability for torrent $h_{\mathcal{T}}$, peer $P_i$ locates the $k$ nodes closest to $h_{\mathcal{T}}$ (typically $k=20$) through iterative lookups and stores its contact information $(\mathsf{pk}_i, \mathsf{ip}_i, \mathsf{port}_i)$ on those nodes. To discover peers, a client performs the same lookup and retrieves peer records from the closest nodes. These operations complete in $O(\log n)$ steps for $n$ participants.

Each \texttt{.torrent} file includes bootstrap information $(\mathcal{B}, \mathsf{addr}_{\mathsf{rep}})$, where $\mathcal{B} = \{(\mathsf{ip}_i, \mathsf{port}_i)\}_{i \in [k]}$ lists DHT bootstrap nodes and $\mathsf{addr}_{\mathsf{rep}}$ specifies the reputation contract address. When the tracker becomes unavailable, peer $P_i$ joins with node identifier $\mathsf{nodeID}_i = \Hash(\mathsf{pk}_i)$, binding each DHT identity to a verifiable on-chain user.

When peer $P_i$ announces for torrent $h_{\mathcal{T}}$, it sends message $m = (\mathtt{announce} \parallel h_{\mathcal{T}} \parallel \mathsf{pk}_i \parallel \mathsf{ip}_i \parallel \mathsf{port}_i)$ with signature $\sigma_i = \Sign(\mathsf{sk}_i, m)$ to the $k$ closest nodes. Each node $n$ verifies the signature and queries the smart contract: $(\mathsf{pk}_i', u_i, d_i) \gets \SCRead(\mathsf{addr}_{\mathsf{rep}}, \mathsf{uid}_i)$. Node $n$ accepts the announcement only if $\mathsf{pk}_i' = \mathsf{pk}_i$ and $\Rep(u_i, d_i) \ge \mathsf{MinRep}$, then adds $(\mathsf{pk}_i, \mathsf{ip}_i, \mathsf{port}_i)$ to its peer list for $h_{\mathcal{T}}$.

Each peer $P_i$ maintains a local view $\mathcal{S}_{\mathsf{local}}^{(i)} \subseteq \mathcal{S}$ of active peers for each torrent. When $P_j$ announces to $P_i$, the protocol mirrors the centralized $\Announce$ procedure (\cref{fig:pbts-announce}), except that $P_i$ uses its local view instead of a global tracker database. $P_i$ verifies $P_j$'s signature $\sigma_j = \Sign(\mathsf{sk}_j, \mathtt{announce} \parallel \mathsf{uid}_j \parallel h_{\mathcal{T}} \parallel \mathsf{event})$, checks on-chain registration and reputation, updates $\mathcal{S}_{\mathsf{local}}^{(i)}$, and returns a random sample $\mathcal{P} \sample \mathcal{S}_{\mathsf{local}}^{(i)}$. Local views converge over time through authenticated DHT announcements, peer exchange (PEX), and periodic re-announcements. Peers can cache contact information from tracker responses during normal operation to build their own bootstrap node sets $\mathcal{B}_{\mathsf{cached}}$, ensuring rapid DHT network joining when the tracker becomes unavailable. By Kademlia's logarithmic routing properties, active peers remain discoverable in $O(\log n)$ hops.

File transfer follows the attestation protocol (see \cref{subsec:construction}). Each piece transfer from $P_s$ to $P_r$ produces receipt $\sigma_{\mathsf{receipt}} = \Attest(\mathsf{sk}_r, \mathsf{pk}_s, p_i, \mathcal{T}, t_{\text{epoch}})$. During tracker downtime, peers store receipts locally in $\mathcal{R} = \{(\mathsf{pk}_j, p_i, \mathcal{T}, \sigma_j)\}$. After recovery, peers submit accumulated receipts via $\Report$ to update on-chain reputation.

\subsection{Optimized attestation}
\label{subsec:optimized-attest}

Per-piece BLS attestation ensures strong verifiability but introduces computational overhead during high-throughput transfers. We discuss several optimizations that reduce this signing overhead.

The frequency of cryptographic attestation can be adjusted based on trust dynamics during a transfer. At session start, trust between peers is not yet established: the sender has not demonstrated reliability and the receiver may defect. Frequent signatures during this phase provide strong accountability and enable early termination if either party misbehaves. As the session progresses and mutual trust builds through successful exchanges, signing frequency can decrease. Near session end, frequent signatures resume: the remaining pieces become highly valuable to the receiver, and the sender requires proof of delivery to claim full credit. A practical policy signs every piece during the first 100 pieces, every 10 pieces during the middle phase, and every piece for the final 100 pieces. For a transfer with $n$ pieces (where $n > 200$), this requires $200 + \lceil(n-200)/10\rceil$ signatures instead of $n$, reducing overhead by approximately $(1 - \frac{200 + (n-200)/10}{n}) \approx 0.9 - \frac{180}{n}$, approaching 90\% reduction for large files while maintaining security at critical trust boundaries.

Established reputable peers can negotiate reduced signing frequencies. The tracker provides reputation scores during the $\Announce$ phase. High-reputation receivers may propose signing every $k$ pieces, where $k$ scales with reputation. Senders accept this proposal only if receivers' on-chain reputation exceeds a threshold. If a receiver later defects, senders report fraud using the partial attestations, and the tracker penalizes the receiver's reputation. This approach amortizes overhead for trusted peers while maintaining accountability through reputation at risk.

Rather than signing individual pieces, receivers can sign commitments to batches of pieces.
A receiver computes a Merkle tree over piece hashes and signs the root after transferring $k$ pieces. The sender verifies each piece against the torrent metadata during transfer and accepts the batch signature as proof of all $k$ pieces. For $k=10$, this reduces signature operations by 10$\times$. The trade-off is reduced granularity: if the receiver defects mid-batch, the sender loses credit for transferred pieces. Batch size should be chosen based on piece value: smaller batches for high-value content, larger batches for bulk transfers.

For scenarios requiring per-piece attestation with minimal overhead, we switch to a more efficient signature scheme for the duration of the transfer.
We employ ECDSA (secp256k1)~\cite{JohnsonMV01}, which offers faster signing and verification than BLS; on most consumer devices ECDSA operations are further accelerated by dedicated cryptographic hardware, making per-piece signing cheap in practice.
Because ECDSA keys are not registered on-chain, this switch requires a handshake that roots the ephemeral ECDSA keypair in the receiver's long-term BLS identity: the receiver generates a fresh ECDSA keypair for the session and certifies it with a single BLS signature, binding the fast per-piece scheme to the overarching reputation scheme.
The tracker can then verify the BLS certificate once and accept all subsequent ECDSA receipts from that session without further BLS operations.

Let $\Sigma_{\mathsf{BLS}}$ be the long-term signature scheme and $\Sigma_{\mathsf{ECDSA}}$ be the ECDSA scheme. At session start, the receiver generates $(\mathsf{sk}^{\mathsf{ECDSA}}, \mathsf{pk}^{\mathsf{ECDSA}})$ and signs $\mathsf{pk}^{\mathsf{ECDSA}}$ along with session metadata using their long-term BLS key, producing a session certificate:
\begin{align*}
\mathsf{sid} &\sample \{0,1\}^{256} \\
m_0 &\gets (\mathsf{sid} \parallel h_{\mathcal{T}} \parallel \mathsf{pk}_{\mathsf{sender}} \parallel \mathsf{pk}^{\mathsf{ECDSA}}) \\
\sigma_0 &\gets \Sigma_{\mathsf{BLS}}.\Sign(\mathsf{sk}_{\mathsf{receiver}}^{\mathsf{BLS}}, m_0)
\end{align*}

For each piece $p_i$, the receiver signs $(h_{\mathcal{T}}, \mathsf{pk}_{\mathsf{sender}}, h_i, i)$ using $\mathsf{sk}^{\mathsf{ECDSA}}$. The sender verifies each signature immediately, preserving tit-for-tat. All per-piece signatures $\{\sigma_i\}_{i \in [n]}$ are collected and reported to the tracker, which verifies each signature individually. The tracker then credits the sender $n \times \mathsf{piece\_size}$ bytes.

When reporting transfers with multiple peers, the sender aggregates BLS session authentication signatures across peers. For sessions with peers indexed by $j \in \mathcal{J}$, each with session certificate $\mathsf{cert}_j = (\mathsf{sid}_j, h_{\mathcal{T}}, \mathsf{pk}_{\mathsf{sender}}, \mathsf{pk}_j^{\mathsf{ECDSA}})$ and ECDSA signatures $\{\sigma_{i,j}\}$ for transferred pieces, the sender computes:
\begin{align*}
\sigma_{\mathsf{agg}} \gets \Sigma_{\mathsf{BLS}}.\Aggregate(\{(\mathsf{pk}_j, \mathsf{cert}_j, \sigma_{0,j})\}_{j \in \mathcal{J}})
\end{align*}
The report contains one aggregated BLS signature plus all ECDSA per-piece signatures for each peer session.

\Cref{tab:attest-comparison} compares optimization approaches for a 5.1 GB torrent with 2,560 pieces (2 MB each).

\begin{table}[htbp]
\centering
\caption{Comparison of attestation optimizations} 
\label{tab:attest-comparison}

\begin{tabular}{l@{\hspace{-5.0pt}}cc@{\hspace{-6pt}}c}
\toprule
\textbf{Approach} & \textbf{Signatures} & \textbf{Sign Time}\textsuperscript{a} & \textbf{Report Size}\textsuperscript{b} \\
\midrule
BLS & 2,560 & 2.16s & 2.3 MB\textsuperscript{c} \\
Adaptive freq. & $\sim$512 & 0.43s & 456 KB \\
Batch ($k=10$) & 256 & 0.22s & 228 KB \\
ECDSA & 2,560 & $<$0.22s\textsuperscript{d} & 160 KB \\
\bottomrule
\multicolumn{4}{l}{\footnotesize \textsuperscript{a}Signing cost only (Rust/blst BLS receipt generation).} \\
\multicolumn{4}{l}{\footnotesize \textsuperscript{b}Report size assumes BLS aggregation where applicable.} \\
\multicolumn{4}{l}{\footnotesize \textsuperscript{c}\parbox[t]{0.9\columnwidth}{After BLS aggregation: 32 bytes per piece hash plus one 48-byte agg. signature.}} \\
\multicolumn{4}{l}{\footnotesize \textsuperscript{d}\parbox[t]{0.9\columnwidth}{ECDSA (secp256k1) signing is ${\approx}10\times$ faster than BLS (0.08 ms vs.\ 0.84 ms per piece).}}
\end{tabular}
\end{table}

ECDSA provides fast per-piece attestation with moderate bandwidth overhead. With signing approximately 10$\times$ faster than BLS (0.08 ms per piece compared to 0.84 ms), ECDSA reduces computational overhead by 90\% while maintaining full per-piece verification during transfer. The report size of 160 KB for 2,560 pieces represents an $\approx$14$\times$ reduction compared to per-piece BLS before aggregation, though larger than batched approaches. Adaptive and batch approaches trade granularity for further reduced computational and bandwidth overhead. 
\section{Security analysis}
\label{sec:security}

We now formally analyze the security properties of our Persistent BitTorrent Tracker System (PBTS). Our security model addresses the core threats identified in \cref{subsec:threat-model}: reputation manipulation, false reporting, and unauthorized tracker operations. Each property is defined through a security game between a challenger and a probabilistic polynomial-time (PPT) adversary $\mathcal{A}$, capturing the adversary's advantage through a probability expression.

This section relies on three assumptions:
\emph{(1)} The signature scheme $\Sigma$ satisfies strong unforgeability (sUF-CMA) and aggregate unforgeability per \cref{def:signature}.
\emph{(2)} The smart contract layer ensures integrity of state transitions, rejecting unauthorized writes. \emph{(3)} In the untrusted-operator model, the TEE provides correct execution with secure attestation. Properties~\ref{subsec:reg-auth}--\ref{subsec:receipt-nonreuse} hold under assumptions (1)--(2) alone in the trusted-operator model; assumption~(3) additionally protects against a malicious tracker operator.

We analyze four security properties. \emph{Registration authenticity} (\cref{subsec:reg-auth}) prevents identity impersonation by requiring valid signatures from key holders during account creation. \emph{Receipt non-repudiation} (\cref{subsec:receipt-nonrepud}) makes it impossible for peers to deny having received data after signing acknowledgments. \emph{Report soundness} (\cref{subsec:report-soundness}) bounds reputation inflation: users can only claim credit supported by valid receipts from other peers. \emph{Receipt non-reusability} (\cref{subsec:receipt-nonreuse}) prevents double-spending attacks where the same receipt appears in multiple reports.

These properties guarantee that reputation scores reflect actual contributions, the on-chain state remains consistent with real file transfers, and malicious users gain no advantage over honest participants.

\subsection{Registration authenticity}
\label{subsec:reg-auth}

The first requirement for a secure reputation system is that identities cannot be forged or stolen. In PBTS, each user registers by proving ownership of a public key through a digital signature. This prevents adversaries from registering under someone else's public key or creating accounts without corresponding secret keys, which would enable various attacks such as reputation theft or Sybil identity creation without accountability.

\begin{definition}[Registration authenticity]
\label{def:reg-auth}
A PBTS scheme satisfies registration authenticity if for any PPT adversary $\mathcal{A}$, the probability
\begin{align*}
\Pr\left[\begin{array}{l}
\mathsf{pp} \gets \Setup(1^\lambda); \\
(\mathsf{uid}^*, \mathsf{pk}^*, \sigma^*, \mathsf{params}^*) \gets \mathcal{A}^{\Register(\cdot)}(\mathsf{pp}) : \\
\Register(\mathsf{uid}^*, \mathsf{pk}^*, \sigma^*, \mathsf{params}^*) = 1 \\
\land\, \mathsf{pk}^* \notin \mathcal{Q}_{\mathsf{reg}}
\end{array}\right]
\end{align*}
is negligible in $\lambda$, where $\mathcal{Q}_{\mathsf{reg}}$ is the public key set $\mathcal{A}$ submitted to $\Register$ queries (i.e., keys for which $\mathcal{A}$ generated or obtained the corresponding secret keys).
\end{definition}

\begin{theorem}
\label{thm:reg-auth}
If $\Sigma$ is an EUF-CMA secure signature scheme and the TEE provides correct execution and attestation, then the PBTS construction from \cref{subsec:construction} satisfies registration authenticity.
\end{theorem}
\begin{proof}
We prove by reduction to the EUF-CMA security of $\Sigma$, which is implied by the assumed sUF-CMA security. Suppose there exists a PPT adversary $\mathcal{A}$ that breaks registration authenticity with non-negligible advantage $\epsilon$. We construct a PPT algorithm $\mathcal{B}$ that uses $\mathcal{A}$ to break the EUF-CMA security of $\Sigma$ with advantage $\epsilon$.

$\mathcal{B}$ receives a challenge public key $\mathsf{pk}^*$ from the EUF-CMA challenger and has access to a signing oracle $\mathcal{O}_{\Sign}(\mathsf{sk}^*, \cdot)$. $\mathcal{B}$ runs $\mathsf{pp} \gets \Setup(1^\lambda)$ and gives $\mathsf{pp}$ to $\mathcal{A}$. 

When $\mathcal{A}$ makes registration queries, $\mathcal{B}$ responds like so:
\begin{itemize}
  \item For queries with $\mathsf{pk} \neq \mathsf{pk}^*$: $\mathcal{B}$ generates independent key pairs $(\mathsf{sk}, \mathsf{pk}) \gets \KeyGen()$ and processes registration normally, recording $\mathsf{pk}$ in $\mathcal{Q}_{\mathsf{reg}}$.
  \item For queries involving $\mathsf{pk}^*$: $\mathcal{B}$ uses its signing oracle $\mathcal{O}_{\Sign}$ to generate the registration signature, but does not add $\mathsf{pk}^*$ to $\mathcal{Q}_{\mathsf{reg}}$ since $\mathcal{B}$ does not know $\mathsf{sk}^*$.
\end{itemize}

When $\mathcal{A}$ outputs $(\mathsf{uid}^*, \mathsf{pk}^*, \sigma^*, \mathsf{params}^*)$ with $\mathsf{pk}^* \notin \mathcal{Q}_{\mathsf{reg}}$ and $\Register(\mathsf{uid}^*, \mathsf{pk}^*, \sigma^*, \mathsf{params}^*) = 1$, the registration algorithm (\cref{fig:pbts-register}) verifies:
\begin{align*}
m^* = (\mathtt{register} \parallel \mathsf{iid} \parallel \mathsf{uid}^*) \\
\Sigma.\Verify(\mathsf{pk}^*, m^*, \sigma^*) = 1
\end{align*}

Since $\mathsf{pk}^* \notin \mathcal{Q}_{\mathsf{reg}}$ and $m^*$ was not queried to $\mathcal{O}_{\Sign}$, the pair $(m^*, \sigma^*)$ constitutes a valid signature forgery. $\mathcal{B}$ outputs this forgery, breaking the EUF-CMA security of $\Sigma$ with advantage $\epsilon$ and contradicting the assumed EUF-CMA security of $\Sigma$, so $\epsilon$ must be negligible. 
\end{proof}

\subsection{Receipt non-repudiation}
\label{subsec:receipt-nonrepud}

A key component of our system is the p2p attestation mechanism, where receiving peers sign cryptographic receipts acknowledging piece transfers. For this to be meaningful, receipts must be non-repudiable: one cannot credibly deny receiving data after producing a valid receipt.

\begin{definition}[Receipt non-repudiation]
\label{def:receipt-nonrepud}
A PBTS scheme satisfies receipt non-repudiation if for any PPT adversary $\mathcal{A}$, the probability
\begin{align*}
\Pr\left[\begin{array}{l}
\mathsf{pp} \gets \Setup(1^\lambda, \mathsf{MinRep}, \mathsf{InitCredit}); \\
(\mathsf{sk}_S, \mathsf{pk}_S) \gets \KeyGen(); \\
(p_i, \mathcal{T}, t_{\text{epoch}}, \mathsf{sk}_\mathcal{A}, \mathsf{pk}_\mathcal{A}) \gets \mathcal{A}(\mathsf{pp}, \mathsf{pk}_S); \\
\sigma_{\mathsf{receipt}} \gets \Attest(\mathsf{sk}_\mathcal{A}, \mathsf{pk}_S, p_i, \mathcal{T}, 
 t_{\text{epoch}}); \\
\mathsf{b} \gets \Report(\mathsf{uid}_S, \mathsf{pk}_S, \{\mathsf{pk}_\mathcal{A}\}, \mathcal{T}, \\
\quad\quad \{t_{\text{epoch}}\}, \sigma_{\mathsf{receipt}}, \Delta_{\mathsf{up}}) : \\
\Verify(\mathsf{pk}_\mathcal{A}, \mathsf{pk}_S, p_i, \mathcal{T}, t_{\text{epoch}}, \\
\quad\quad \sigma_{\mathsf{receipt}}) = 1 \land \mathsf{b} = 0
\end{array}\right]
\end{align*}
is negligible in $\lambda$. The adversary wins if $\Attest$ produces a receipt $\sigma_{\mathsf{receipt}}$ that verifies but causes an honest sender's $\Report$ to fail (successful repudiation).
\end{definition}

\begin{theorem}[Receipt non-repudiation]
\label{thm:receipt-nonrepud}
If $\Sigma$ is EUF-CMA secure and the smart contract enforces access control, then the PBTS construction satisfies receipt non-repudiation.
\end{theorem}

\begin{proof}
Suppose adversary $\mathcal{A}$ wins the non-repudiation game with non-negligible advantage $\epsilon$. Then $\mathcal{A}$ produces a receipt $\sigma_{\mathsf{receipt}} = \Attest(\mathsf{sk}_\mathcal{A}, \mathsf{pk}_S, p_i, \mathcal{T}, t_{\text{epoch}})$ that satisfies $\Verify(\mathsf{pk}_\mathcal{A}, \mathsf{pk}_S, p_i, \mathcal{T}, t_{\text{epoch}}, \sigma_{\mathsf{receipt}}) = 1$, but the honest sender's $\Report$ call returns $0$.

The $\Report$ algorithm (\cref{fig:pbts-report-timestamp}) rejects a report only if signature verification fails, the timestamp $t_{\text{epoch}}$ is outside the valid epoch window $[t_{\text{now}} - \Delta, t_{\text{now}}]$, the receipt has been used before (double-spending with $\mathsf{rid} \in \mathcal{R}_{\text{recent}}$), or the piece hash does not match ($h_i \neq \Hash(p_i)$).

By the winning condition, signature verification succeeds, so the first condition does not hold. For an honest sender immediately reporting a new transfer, $t_{\text{epoch}}$ is current and within the valid window. The receipt is used for the first time, so $\mathsf{rid} \notin \mathcal{R}_{\text{recent}}$. The honest sender uses the actual $p_i$ transferred by $\mathcal{A}$, so $h_i = \Hash(p_i)$ holds by construction.

Since none of the rejection conditions hold, the algorithm must return $1$, contradicting the assumption that $\mathcal{A}$ wins with $\Report$ returning $0$. Therefore, $\epsilon$ must be negligible.
\end{proof}

\subsection{Report soundness}
\label{subsec:report-soundness}

With registration authenticity and receipt non-repudiation established, we now address the core security property: users cannot inflate reputation beyond their actual contributions. This property prevents false reporting by requiring that every transfer be backed by a cryptographic acknowledgment from the counterparty.

\begin{definition}[Report soundness]
\label{def:report-soundness}
A PBTS scheme satisfies report soundness if for any PPT adversary $\mathcal{A}$ that interacts with honest peers and the tracker, the probability
\begin{align*}
\Pr\left[\begin{array}{l}
\mathsf{pp} \gets \Setup(1^\lambda); \\
(\mathsf{uid}_\mathcal{A}, \mathsf{pk}_\mathcal{A}) \gets \mathcal{A}^{\Register(\cdot), \Announce(\cdot), \text{Peers}}(\mathsf{pp}); \\
(\mathsf{up}_{\text{true}}, \mathsf{down}_{\text{true}}) \gets \mathsf{TrueStats}(\mathcal{A}); \\
\Report(\mathsf{uid}_\mathcal{A}, \mathsf{pk}_\mathcal{A}, \{\mathsf{pk}_j\}, \mathcal{T}, \{t_j\}, \sigma_{\mathsf{agg}}, \Delta_{\mathsf{up}}) = 1; \\
(\mathsf{up}_{\text{claimed}}, \mathsf{down}_{\text{claimed}}) \gets \SCRead(\mathsf{addr}, \mathsf{uid}_\mathcal{A}) : \\
\mathsf{up}_{\text{claimed}} > \mathsf{up}_{\text{true}} \lor \mathsf{down}_{\text{claimed}} < \mathsf{down}_{\text{true}}
\end{array}\right]
\end{align*}
is negligible in $\lambda$, where $\mathsf{TrueStats}(\mathcal{A})$ tracks actual data uploaded and downloaded by $\mathcal{A}$ to/from honest peers.
\end{definition}

\begin{theorem}
\label{thm:report-soundness}
If $\Sigma$ satisfies aggregate unforgeability, the TEE executes the tracker code correctly, and the smart contract enforces access control, then the PBTS construction satisfies report soundness.
\end{theorem}

\begin{proof}
The $\Report$ algorithm (\cref{fig:pbts-report-timestamp}) accepts a report only if the aggregated signature $\sigma_{\mathsf{agg}}$ verifies over the receipt set $\{(\mathsf{pk}_j, m_j)\}_{j \in \mathcal{J}}$, where each $m_j = (h_{\mathcal{T}} \parallel \mathsf{pk}_\mathcal{A} \parallel h_j \parallel j \parallel t_j)$ represents a receipt from peer $j$ acknowledging receipt of piece $j$ from $\mathcal{A}$ at time $t_j$.
Consider two cases:

\paragraphNoSkip{Case 1: Over-reporting uploads}
% ($\mathsf{up}_{\text{claimed}} > \mathsf{up}_{\text{true}}$)}
% 
For $\mathcal{A}$ to claim credit exceeding its actual uploads, it must provide valid receipts for pieces it never sent. This requires forging receipts from honest peers who never signed them. 
Let $\mathcal{B}$ be a PPT algorithm attacking aggregate unforgeability. $\mathcal{B}$ receives challenge public keys $\{\mathsf{pk}_1^*, \ldots, \mathsf{pk}_k^*\}$ for $k$ honest peers and access to signing oracles $\{\mathcal{O}_{\Sign}(\mathsf{sk}_i^*, \cdot)\}_{i \in [k]}$. $\mathcal{B}$ simulates the PBTS environment for $\mathcal{A}$, using the challenge keys as honest peers' public keys and the signing oracles to generate legitimate receipts when $\mathcal{A}$ uploads to honest peers.
When $\mathcal{A}$ submits a report claiming $\mathsf{up}_{\text{claimed}} > \mathsf{up}_{\text{true}}$, the report includes an aggregate signature $\sigma_{\mathsf{agg}}$ over receipts $\{(\mathsf{pk}_j, m_j)\}_{j \in \mathcal{J}}$. Since $\mathcal{A}$ over-reported uploads, at least one $(\mathsf{pk}_j^*, m_j^*)$ must correspond to an upload that never occurred, meaning $m_j^*$ was never queried to $\mathcal{O}_{\Sign}(\mathsf{sk}_j^*, \cdot)$. Thus $\sigma_{\mathsf{agg}}$ constitutes a forgery, and $\mathcal{B}$ outputs it to break aggregate unforgeability with the same advantage as $\mathcal{A}$.

\paragraphNoSkip{Case 2: Under-reporting downloads}
% ($\mathsf{down}_{\text{claimed}}<\mathsf{down}_{\text{true}}$)
% 
When $\mathcal{A}$ downloads piece $p_i$ from an honest peer $P$, it must generate $\sigma_{\mathsf{receipt}} = \Attest(\mathsf{sk}_\mathcal{A}, \mathsf{pk}_P, p_i, \mathcal{T}, t_{\text{epoch}})$ and return it to $P$. If $\mathcal{A}$ refuses to do so, honest peers detect non-cooperation and stop serving $\mathcal{A}$ (tit-for-tat enforcement). If $\mathcal{A}$ provides receipts, those can be submitted by honest uploaders, accurately tracking $\mathcal{A}$'s downloads.
Combining both cases, by $\Sigma$'s aggregate unforgeability:
% \begin{align*}
% \[
$
\mathsf{Adv}^{\mathsf{report}}_{\mathsf{PBTS},\mathcal{A}}(\lambda) \leq \mathsf{Adv}^{\mathsf{agg-forge}}_{\Sigma,\mathcal{B}}(\lambda) 
= \mathsf{negl}(\lambda)
$
% \]
% \end{align*}
\end{proof}

\subsection{Receipt non-reusability}
\label{subsec:receipt-nonreuse}

Even with report soundness ensuring that receipts correspond to genuine transfers, another threat remains: receipt reuse. An adversary might attempt to submit the same receipt multiple times across different reports, effectively claiming credit for the same upload repeatedly.

\begin{definition}[Receipt non-reusability]
\label{def:receipt-nonreuse}
A PBTS scheme satisfies receipt non-reusability if for any PPT adversary $\mathcal{A}$, the following probability is negligible in $\lambda$:
\begin{align*}
\Pr\left[\begin{array}{l}
\mathsf{pp} \gets \Setup(1^\lambda); \\
\sigma_{\mathsf{receipt}} \gets \mathcal{A}^{\Register(\cdot), \Announce(\cdot), \Report(\cdot)}(\mathsf{pp}); \\
(R_1, R_2) \gets \mathcal{A}(\sigma_{\mathsf{receipt}}) : \\
\sigma_{\mathsf{receipt}} \in R_1 \land \sigma_{\mathsf{receipt}} \in R_2 \\
\land\, \Report(R_1) = 1 \land \Report(R_2) = 1
\end{array}\right]
\end{align*}
\end{definition}

\begin{theorem}
\label{thm:receipt-nonreuse}
If $\Sigma$ is sUF-CMA secure and receipts include timestamps, then the PBTS construction satisfies receipt non-reusability.
\end{theorem}

\begin{proof}
Each receipt includes an epoch identifier in the signed message, that is, if $t_{\text{epoch}}$ represents the period during which the transfer occurred, we have:
$$m = (h_{\mathcal{T}} \parallel \mathsf{pk}_{\text{sender}} \parallel h_i \parallel i \parallel t_{\text{epoch}})$$
Time is divided into discrete epochs (e.g., hour-long windows), and $t_{\text{epoch}} = \lfloor t_{\text{current}} / W \rfloor$ where $W$ is the epoch width and $t_{\text{current}}$ is the time when the receipt is generated. The epoch timestamp is fixed at receipt generation: honest receivers sign the current epoch, so $t_{\text{epoch}}$ cannot be greater than the time $t_1$ when the receipt first appears in a report. The epoch timestamp is provided as input to the $\Attest$ algorithm and signed by the receiver.

The tracker enforces two mechanisms: (1) temporal validity, accepting only receipts with recent timestamps, and (2) short-term deduplication via a set $\mathcal{R}_{\text{recent}}$ of recently used receipt identifiers. Each receipt is uniquely identified by $(h_{\mathcal{T}}, \mathsf{pk}_{\text{sender}}, \mathsf{pk}_{\text{receiver}}, h_i, i, t_{\text{epoch}})$.

Say adversary $\mathcal{A}$ successfully reuses receipt $\sigma_{\mathsf{receipt}}$ with time $t_{\text{epoch}}$ by getting it accepted in reports $R_1$ and $R_2$ processed at times $t_1$ and $t_2$ where $t_1 < t_2$. Let the acceptance window at time $t$ cover epochs in range $[t - \Delta, t]$.

\paragraphNoSkip{Case 1: $t_2 \leq t_1 + \Delta + 1$}
The reports are within $\Delta + 1$ epochs of each other. After $R_1$ is processed at epoch $t_1$, the receipt identifier is added to $\mathcal{R}_{\text{recent}}$. Since $t_2 \leq t_1 + \Delta + 1$, the receipt remains in $\mathcal{R}_{\text{recent}}$ when $R_2$ is processed. The deduplication check detects the reuse and rejects $R_2$.

\paragraphNoSkip{Case 2: $t_2 > t_1 + \Delta + 1$}
For the receipt to be valid at time $t_2$, we need $t_{\text{epoch}} \geq t_2 - \Delta$. However, since the receipt was valid at time $t_1$, we have $t_{\text{epoch}} \leq t_1 < t_2 - \Delta - 1 < t_2 - \Delta$. This contradicts the requirement for acceptance at $t_2$, so the receipt is rejected as expired.

The adversary cannot forge receipts with future timestamps because honest receivers sign the current timestamp at generation time. Forging such a signature, or producing a distinct valid signature on an already-submitted receipt message, contradicts the sUF-CMA security of $\Sigma$. Similarly, modifying the timestamp in an existing receipt invalidates the signature.
Since both cases lead to rejection:
$\mathsf{Adv}^{\mathsf{reuse}}_{\mathsf{PBTS},\mathcal{A}}(\lambda) = 0$
\end{proof}

\section{Privacy and federation}
\label{sec:privacy}

Publishing reputation on a public blockchain brings verifiability and portability, but it also introduces privacy trade-offs that do not exist in a purely centralized tracker.
In a classical private tracker, reputation data sits in a private database controlled by the operator; in PBTS it is visible to anyone who can read the chain.
This section analyzes the resulting privacy risks and discusses mechanisms to mitigate them.
We also note that on-chain participation can be made opt-in by tracker operators: a tracker may allow users who do not wish to publish their public key on the smart contract to continue using it in classical mode, forfeiting the benefits of portable reputation, migration, and authenticated DHT fallback.

\paragraphNoSkip{Pseudonymity and its limits}
On-chain reputation records store only user IDs, public keys, and transfer counters.
No plaintext identifiers are published, so there is no direct link to real-world identities.
However, this anonymity is shallow.
Public keys are long-lived linkable identifiers: anyone who correlates a public key with an IP address (e.g., by observing the peer list returned by an $\Announce$ call, which contains $(pk, \mathsf{ip}, \mathsf{port})$ tuples) can permanently de-anonymize the corresponding user.
Once this link is established, the attacker can follow the user's activity across all torrents on the same tracker instance.
A user who reuses the same key across multiple tracker instances is additionally linkable across those instances.
Users should therefore rerandomize their public key for each tracker instance and avoid reusing keys from other contexts (e.g., cryptocurrency wallets or other identity systems), as any external correlation of the key immediately collapses the pseudonymity offered by on-chain storage.
\paragraphNoSkip{Swarm privacy}
Peers learn the mapping between on-chain identities and IP addresses when receiving the swarm member list.
We propose replacing long-term public keys in peer lists with short-lived ephemeral keys.
When joining a swarm, a peer generates an ephemeral key pair $(\mathsf{sk}_e, \mathsf{pk}_e)$ and authenticates to the tracker using its long-term key $\mathsf{sk}$, proving that $\Rep(\mathsf{up}, \mathsf{down}) \geq \mathsf{MinRep}$.
The tracker verifies the reputation, issues a signed session credential $\tau = \Sign(\mathsf{sk}_{\mathsf{tracker}}, \mathsf{pk}_e \parallel \mathsf{tid} \parallel t_{\mathsf{epoch}})$, and records $(\mathsf{pk}_e, \mathsf{ip}, \mathsf{port})$ in the swarm instead of the long-term key.
The mapping between the long-term key and $\mathsf{pk}_e$ is maintained privately by the tracker.
Peers receiving the list see only ephemeral keys; without the tracker's internal state, the link to on-chain identities cannot be reconstructed.
Receipts are signed with $\mathsf{sk}_e$; the peer privately provides the tracker with the binding $\mathsf{pk}_e \to \mathsf{uid}$ so that download credit is attributed to the correct account.
Note that this mechanism hides the long-term identity from other peers in the swarm, but the tracker itself must still identify the reporting peer: $\Report$ credits a specific account, so the tracker needs to know the binding $\mathsf{pk}_e \to \mathsf{uid}$.
A ZK proof that hides the account identity entirely from the tracker would break this attribution.
Under the trusted-operator model, the tracker maintains this binding in its private state.
Under the untrusted-operator model, the binding is processed inside the TEE enclave and never exposed to the operator.
In both cases, the ephemeral key scheme provides peer-facing privacy; the models differ only in whether the operator is trusted to keep the binding confidential.

\paragraphNoSkip{Confidential reputation}
A further limitation of the base design is that reputation counters are stored in plaintext on the blockchain: any observer can read the exact upload and download totals for any registered public key.
To address this, reputation values can be stored as commitments rather than cleartext.
Using a homomorphic commitment scheme (e.g., Pedersen commitments over an appropriate group), the contract stores $(C_{\mathsf{up}}, C_{\mathsf{down}})$ where $C_{\mathsf{up}} = \mathsf{Com}(\mathsf{up}; \rho_{\mathsf{up}})$ and $C_{\mathsf{down}} = \mathsf{Com}(\mathsf{down}; \rho_{\mathsf{down}})$ for randomness $(\rho_{\mathsf{up}}, \rho_{\mathsf{down}})$ known only to the client.

When submitting a $\Report$ claiming $\Delta_{\mathsf{up}}$ uploads, the client computes a new commitment $C_{\mathsf{up}}'$ to the updated value and provides a ZK proof $\pi$ that the update is consistent with the receipts:
\begin{align*}
  \pi \colon \exists\, \mathsf{up}, \rho, \rho' :\;&\;
  C_{\mathsf{up}} = \mathsf{Com}(\mathsf{up};\, \rho) \\
  \land\;&\;
  C_{\mathsf{up}}' = \mathsf{Com}(\mathsf{up} + \Delta_{\mathsf{up}};\, \rho').
\end{align*}
The tracker verifies $\pi$ together with the receipt aggregate signature, then writes $C_{\mathsf{up}}'$ to the contract.
No cleartext counter is written to the chain; correctness is enforced by the proof rather than by trusting the tracker's arithmetic.

To participate in a reputation-gated operation, the client constructs a ZK proof of the form
\begin{align*}
  \exists\, v, \rho, v', \rho' :\;&\;
  \mathsf{Com}(v;\rho) = C_{\mathsf{up}} \\
  \land\;&\; \mathsf{Com}(v';\rho') = C_{\mathsf{down}} \\
  \land\;&\; \Rep(v, v') \geq \mathsf{MinRep},
\end{align*}
without revealing exact counter values.
This provides confidentiality: observers see only commitments, while each client retains the opening of their own commitment and can produce proofs independently of the tracker.

\paragraphNoSkip{Plausible deniability}
A tracker operator wishing to offer its members some degree of plausible deniability could register decoy accounts using valid public keys sampled from the broader ledger (e.g., active addresses on the same blockchain), making it harder to assert with certainty that a given key belongs to an actual tracker member.
This raises the uncertainty for an observer trying to identify members from the on-chain allow-list alone.
Decoy entries are non-functional by construction: since reputation is stored as a commitment, participating in any reputation-gated operation requires proving knowledge of the opening $(v, \rho)$.
When registering a decoy, the tracker generates a commitment with randomness it immediately discards; no one can open it, so the entry confers no usable credential.
The consent concern remains: a third party's key is registered on a public ledger without their knowledge, although no false reputation or usable access is granted to them.

\paragraphNoSkip{DHT context}
When no tracker is reachable, peers fall back to the authenticated DHT for peer discovery (see \cref{subsec:dht-fallback}).
This setting offers strictly stronger privacy guarantees than the tracker-mediated case because verification is one-way: a DHT node only needs to confirm that an announcing peer is a legitimate member of the allow-list with sufficient reputation, without needing to credit any account.
This makes it possible to use a zero-knowledge membership proof: a peer proves that its key appears in the smart contract allow-list and that $\Rep(\mathsf{up}, \mathsf{down}) \geq \mathsf{MinRep}$, without revealing which entry it corresponds to.
No binding to specific accounts is required, so the proof fully decouples DHT participation from on-chain identity.
Reputation updates are deferred in this setting: since \texttt{SCWrite} is access-controlled to the tracker key, peers cannot update the contract directly.
Receipts are still collected peer-to-peer as usual via $\Attest$/$\Verify$, and peers hold them locally until a successor tracker comes online; at that point, they submit their accumulated receipts to the new instance, which processes them and writes the updated commitment to the contract.
Note that regular receipt signatures carry the signer's public key, which leaks identity when the deferred batch is submitted; a ZK attestation replacing the signature with a proof of knowledge and a per-transfer nullifier would preserve anonymity while still preventing double-spending, albeit at significant proving overhead per piece.
Similarly, IP addresses of participating peers remain visible to other DHT nodes, and hiding them would require external tools such as Tor.

\paragraphNoSkip{Federation}
The same ZK membership proof that enables anonymous DHT participation extends naturally to federated tracker deployments.
Federated tracker instances can each maintain their own reputation contract while mutually recognizing each other's members.
A peer registered on tracker $T_1$ can present a zero-knowledge proof of membership to tracker $T_2$ without revealing which instance they belong to or migrating their key.
Cross-tracker linkability is further prevented by rerandomizing the public key per tracker instance.
Federated instances inherit reputation through the referrer mechanism on migration, preserving continuity while each instance retains sovereignty over its own access control policy.

\section{Implementation and evaluation}
\label{sec:implementation}
\label{sec:evaluation}

The core components are implemented in Rust. Receipt generation and verification are exposed as BEP~10 protocol extensions, ensuring backward compatibility with existing BitTorrent clients. Smart contract interaction is handled by the \texttt{alloy} library. We first discuss two practical deployment concerns, then evaluate performance through micro-benchmarks and system-level simulations.

\subsection{Tracker liveness and duplication}
\label{sec:liveness}

Standard TEE deployments bind keys to hardware, making failover difficult. PBTS decouples key derivation from physical hardware via a decentralized, attested Key Management Service (KMS). Each tracker instance obtains its \textit{Tracker Root Key} ($key_{\mathrm{trk}}$), derived from the cryptographic measurement of the tracker binary and configuration:

\begin{enumerate*}[label=(\roman*)]
    \item The instance presents its TEE attestation to the KMS, which verifies the measurement against an allowlist of approved binaries.
    \item The KMS derives and returns $key_{\mathrm{trk}}$ for this code/configuration.
    \item All subsequent ephemeral keys (contract signing, state encryption, peer authentication) are derived from $key_{\mathrm{trk}}$.
\end{enumerate*}
\noindent

This yields three benefits.
\emph{Stateless recovery}: a replacement instance presents the same attestation, retrieves $key_{\mathrm{trk}}$, and resumes with full on-chain credentials.
\emph{Liveness via duplication}: multiple instances sharing $key_{\mathrm{trk}}$ operate in parallel for load balancing and redundancy.
\emph{Censorship resistance}: since no state is hardware-bound, taking down individual nodes does not disrupt the service.

\subsection{Secure blockchain RPC and state integrity}
\label{sec:blockchainrpc}

Tracker instances require \textit{liveness} (reads and writes to the smart contract always succeed) and \textit{integrity} (blockchain responses are authentic) with respect to blockchain interaction. We address these with three measures: (1) a prioritized list of fallback RPC endpoints across providers; (2) state reads verified against Merkle proofs over block headers, removing reliance on a single RPC provider; (3) contract writes signed with $key_{\mathrm{trk}}$-derived keys, with the contract enforcing that only attested instances can mutate state. Writes are treated as confirmed only after monitoring finality to handle chain reorganizations.

We now evaluate performance through micro-benchmarks and system-level simulations, addressing three questions: \emph{(1)} What is the cost of the core cryptographic operations? \emph{(2)} What is the end-to-end impact on client download throughput? \emph{(3)} What does the hybrid ECDSA/BLS scheme gain, and what are the on-chain costs?

\subsection{Experiment setup}
\label{subsec:setup}

Cryptographic benchmarks (BLS keypair generation, receipt signing and verification) use a Rust implementation with the \texttt{blst} library for BLS12-381~\cite{bls12-381}.
TEE experiments run on a confidential virtual machine hosted on Phala~\cite{phala}, equipped with 2 vCPUs, 4 GB of RAM, and Intel TDX.
Gas costs are measured using a local Ethereum node (Anvil).
Client download simulations vary download speed (1--20 MB/s) and piece size (256 KB and 2 MB); the attestation comparison uses a representative 5.1 GB torrent with 2,560 pieces of 2 MB each.
We open-sourced our implementation and made it available online.\footnote{\url{https://anonymous.4open.science/r/pbts-75AC/}}

\subsection{Micro-benchmarks}
\label{subsec:microbenchmarks}

We measure the baseline costs of the cryptographic operations that form the building blocks of our system. \Cref{tab:microbenchmarks} summarizes the latency of key operations.

\begin{table}[ht]
    \centering
    \caption{Micro-benchmark results (mean latency).}
    \label{tab:microbenchmarks}
    \begin{tabular}{lr}
        \toprule
        \textbf{Operation} & \textbf{Time (ms)} \\
        \midrule
        \multicolumn{2}{l}{\textit{Cryptographic Primitives (Rust/blst)}} \\
        BLS Keypair Generation & 0.22 \\
        Receipt Creation (Sign) & 0.84 \\
        Receipt Verification & 2.28 \\
        \midrule
        \multicolumn{2}{l}{\textit{TEE Operations (Intel TDX)}} \\
        Tracker Key Derivation (No TEE) & 0.21 \\
        Tracker Key Derivation (TEE) & 1.06$^\dagger$ \\
        Attestation Generation & 3.93 \\
        Attestation Verification & 435.25 \\
        \bottomrule
        \multicolumn{2}{l}{\footnotesize $^\dagger$Median reported; distribution is right-skewed due to occasional} \\
        \multicolumn{2}{l}{\footnotesize TDX scheduling delays.} \\
    \end{tabular}
\end{table}

\paragraphNoSkip{Receipt operations}
Using the Rust/blst implementation, receipt creation (BLS signing) takes 0.84 ms and verification takes 2.28 ms per receipt. These per-operation costs are low; the dominant overhead in practice comes from the receipt exchange round-trip rather than computation, as discussed in \cref{subsec:client-perf}.

\paragraphNoSkip{TEE overhead}
Running code inside the TEE introduces measurable overhead. Tracker key derivation inside the TDX enclave takes roughly 1.1 ms (median), compared to 0.2 ms outside (a $\approx5\times$ increase). This cost is paid once at startup and is not a bottleneck for user registration. Attestation generation is efficient ($\approx 4$ ms), while verification takes $\approx 435$ ms, a one-time cost paid during tracker registration or migration, not per-transaction. This latency includes the network round-trip to the Intel Attestation Service (IAS) for quote verification.

\subsection{Client download performance}
\label{subsec:client-perf}

To understand the impact on user experience, we simulated file downloads under various network conditions and piece sizes. The primary metric is \emph{throughput reduction}: the percentage loss in effective download speed due to the time spent generating and exchanging receipts.

\begin{figure}[ht]
    \centering
    \includegraphics[width=0.85\linewidth]{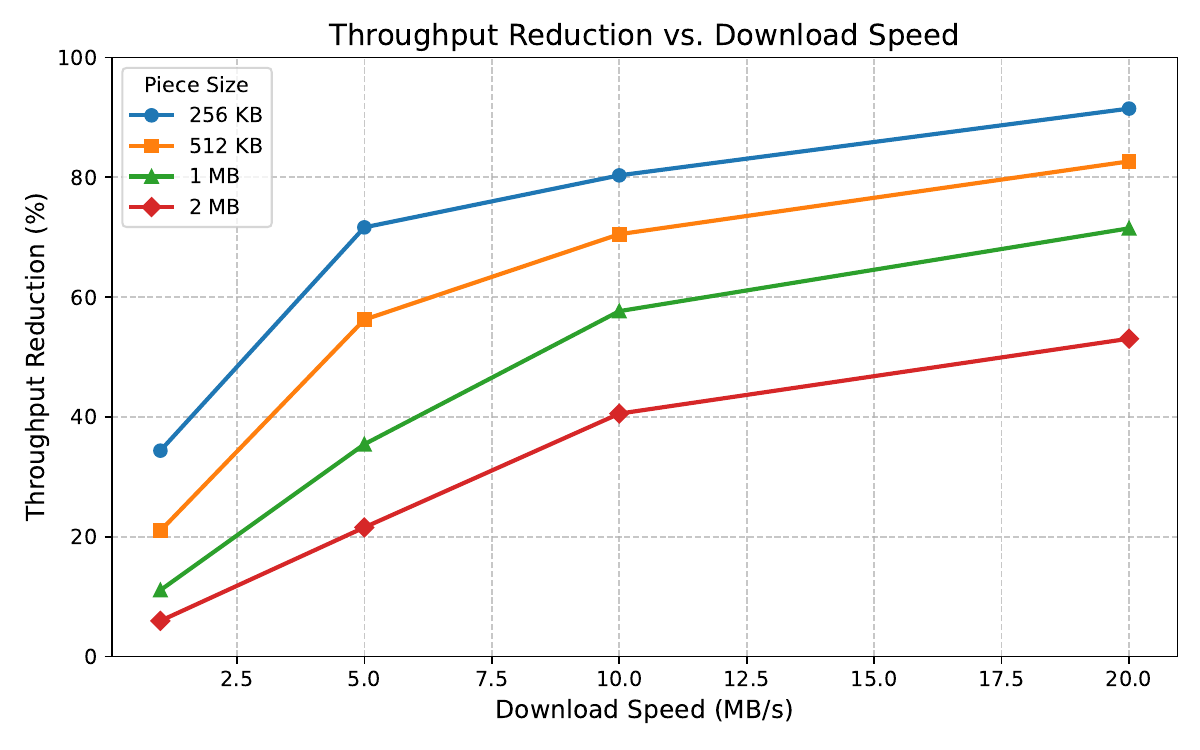}
    \caption{Throughput reduction as a function of download speed under different piece sizes (Batch Size = 1). Smaller pieces incur higher overhead due to more frequent receipt exchanges.}
    \label{fig:throughput_overhead}
\end{figure}

\paragraphNoSkip{Concrete example}
Consider a user downloading a 1 GB file at 1 MB/s. With 256 KB pieces, the client generates 4096 receipts; at 0.84 ms each, local signing takes 3.4 s out of a 1000 s baseline. With 2 MB pieces, only 512 receipts are needed, reducing signing cost to 0.43 s. In both cases cryptographic cost alone is under 0.5\%, yet the observed end-to-end overhead reaches 4.2\% for 2 MB pieces (\cref{fig:throughput_overhead}), meaning the remaining overhead comes from the receipt exchange round-trip with the remote peer.

\paragraphNoSkip{Impact of network speed}
As download speed increases, pieces arrive faster and the per-piece receipt exchange accounts for a growing fraction of transfer time.
At 20 MB/s with 256 KB pieces, a piece arrives every 12.8 ms; thus each receipt round-trip constrains throughput.
Larger piece sizes reduce how often receipts must be exchanged and keep overhead manageable.
For peers operating at high bandwidth with small pieces, the hybrid ECDSA/BLS scheme (see \cref{subsec:analytical-opt}) lowers piece signing cost and thus the time spent on each exchange.

\subsection{Hybrid ECDSA/BLS scheme}
\label{subsec:analytical-opt}

The key tension in PBTS is that BLS supports aggregation (making it ideal for batch verification at report time) but is too slow for per-piece signing during live transfers, where each piece exchange adds latency directly visible to the user. ECDSA is roughly $10\times$ faster per operation but is not aggregatable. The hybrid scheme resolves this by assigning each primitive to the role it is best suited for: ECDSA for per-piece attestation during transfer, where each exchange begins with a handshake that roots the ephemeral ECDSA keypair in the receiver's long-term BLS identity, and BLS for aggregate verification when the uploader reports to the tracker.

At report time, BLS aggregation reduces the tracker's verification cost to a single pairing operation: 12.57 ms for 100 receipts, an 18$\times$ speedup over individual verification ($\approx$226 ms, see \Cref{tab:agg-speedup}). The speedup grows with batch size, reaching 22.6$\times$ at 500 receipts (\Cref{tab:agg-speedup}). Per-piece ECDSA signing during transfer costs $\approx$0.08 ms per piece, keeping the per-exchange overhead minimal.

\begin{table}[ht]
    \centering
    \caption{BLS aggregate verification speedup vs.\ individual verification (Rust/blst).}
    \label{tab:agg-speedup}
    \begin{tabular}{rrrr}
        \toprule
        \textbf{Batch} & \textbf{Aggregate (ms)} & \textbf{Individual (ms)} & \textbf{Speedup} \\
        \midrule
        10  &  3.78  &   22.77  &  6.0$\times$ \\
        25  &  5.85  &   56.79  &  9.7$\times$ \\
        50  &  8.04  &  116.08  & 14.4$\times$ \\
        100 & 12.57  &  226.29  & 18.0$\times$ \\
        500 & 49.90  & 1126.75  & 22.6$\times$ \\
        \bottomrule
    \end{tabular}
\end{table}

\Cref{tab:attest-comparison} quantifies the hybrid scheme for a 5.1 GB torrent. Using BLS for every piece costs 2.16 s of signing time; switching to ECDSA per-piece reduces this to under 0.22 s while preserving per-piece security guarantees. \Cref{fig:projected_optimizations} projects the throughput overhead across file sizes: the hybrid scheme brings overhead close to that of adaptive-frequency or batch-signing strategies, while maintaining stronger per-piece verifiability. Under concurrent load, tracker-side report processing latency remains bounded at 5.6--8.5 ms; since reports are submitted at epoch boundaries, this is only relevant for peers with borderline reputation who need an immediate update.

\begin{figure}[ht]
    \centering
    \includegraphics[width=0.85\linewidth]{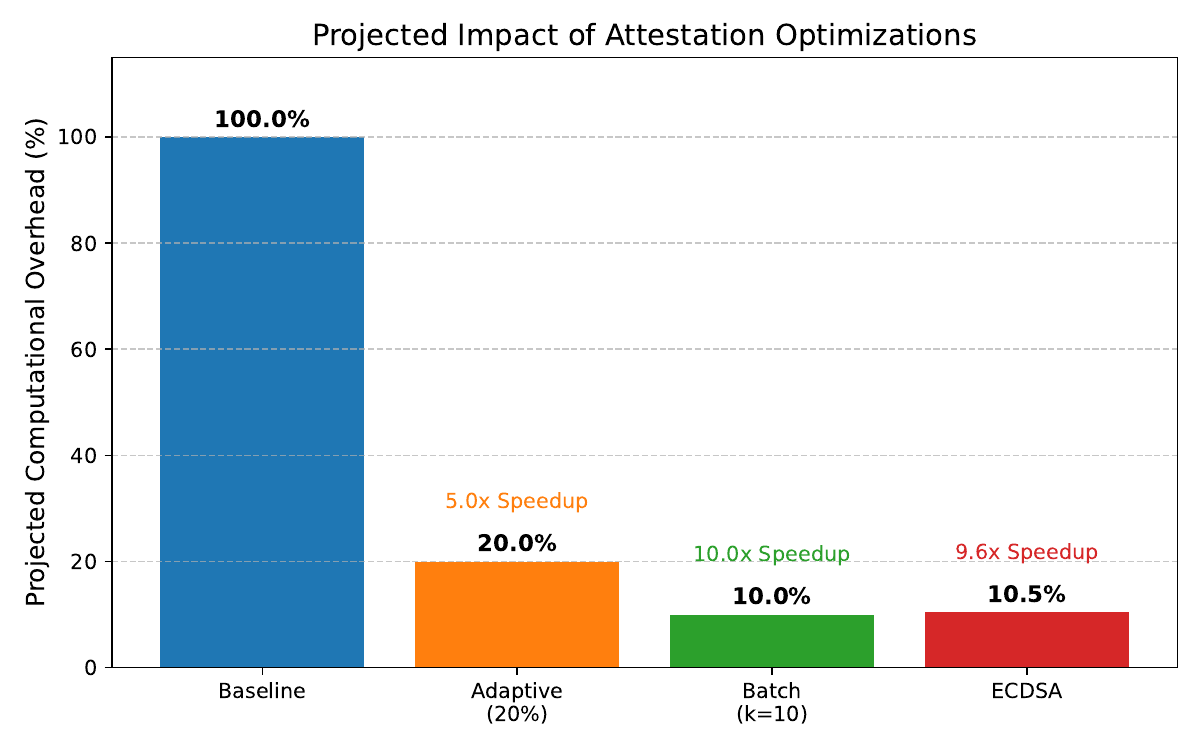}
    \caption{Projected signing overhead relative to baseline (per-piece BLS) across file sizes. The hybrid ECDSA/BLS scheme matches or outperforms batching strategies while preserving per-piece attestation.}
    \label{fig:projected_optimizations}
\end{figure}

\subsection{Smart contract gas costs}
\label{subsec:gas-costs}

We measured gas consumption for each PBTS smart contract operation. \Cref{tab:gas-costs} reports the results.

\begin{table}[ht]
    \centering
    \caption{Smart contract gas costs and observed latency.}
    \label{tab:gas-costs}
    \begin{tabular}{lrr}
        \toprule
        \textbf{Operation} & \textbf{Gas} & \textbf{Latency (ms)} \\
        \midrule
        Create reputation contract & 1{,}270{,}419 & 31.1 \\
        Add user (\texttt{addUser})     &   119{,}522 & 22.1 \\
        Update reputation (\texttt{updateUser}) &    55{,}715 & 19.9 \\
        Migrate user data               &   157{,}749 & 19.4 \\
        \bottomrule
    \end{tabular}
\end{table}

\paragraphNoSkip{Migration cost and failover}
Tracker migration is triggered by deploying a new reputation contract that references a preceding one (\texttt{migrateUserData}).
The 157{,}749 gas migration call is a one-time fee per tracker transition: at current Ethereum L1 conditions ($\approx$0.14 gwei, $\approx$\$2{,}055/ETH) this amounts to $\approx$\$0.045 per transition. No ledger data is lost: the old contract remains accessible as a read-only predecessor. After the new contract is live, peers can immediately re-register and resume reporting; the efficient DHT bootstrap ensures peer discovery resumes within seconds of the tracker restarting.

\paragraphNoSkip{Economic considerations}
Contract creation and user registration (\texttt{addUser}) are one-time costs.
The main recurring expense is reputation updates (\texttt{updateUser}, 55{,}715 gas).
Because the tracker batches updates at epoch boundaries, on-chain cost scales with update frequency, not transfer volume.
Since updates can be scheduled for low-traffic periods, the effective update cost is low: at current Ethereum L1 conditions ($\approx$0.14 gwei, $\approx$\$2{,}055/ETH), one update costs $\approx$\$0.016. For a community of 1{,}000 users updating weekly, this amounts to under \$1{,}000/year on L1. Deploying on a Layer-2 (e.g., Optimism, Arbitrum) reduces costs by a further 10--100$\times$.
Overall, on-chain costs scale linearly with community size and update frequency, remaining viable even for large communities.
Peers requiring an out-of-schedule update (for instance, ahead of a tracker migration) can pay the gas cost directly; the tracker submits the transaction on their behalf.

\section{Conclusion}
\label{section:Conclusion}
Private BitTorrent trackers suffer from three structural weaknesses: reputation is siloed and lost when a tracker shuts down, centralized servers are single points of failure for both reputation storage and peer discovery, and transfer statistics are self-reported and unverifiable.
PBTS addresses each in turn: factory-deployed smart contracts make reputation portable and resilient to takedowns; an authenticated DHT fallback maintains peer discovery when no tracker is reachable; and cryptographic receipts aggregated via BLS signatures replace self-reported statistics with an auditable on-chain record.
We formalize PBTS, prove four security properties under standard cryptographic assumptions, analyze the privacy risks introduced by anchoring identities on a public ledger and propose several mitigations, and demonstrate its practicality: receipt signing adds under 0.5\% local overhead with end-to-end throughput loss below 5\%; BLS aggregation reduces tracker-side verification cost by up to 22.6$\times$; and on-chain updates remain economical when batched at epoch boundaries.
More broadly, PBTS demonstrates that lightweight blockchain integration can retrofit accountability and resilience into existing peer-to-peer protocols without requiring participants to abandon familiar client software or trust a central authority.

\printbibliography[heading=bibintoc]

\appendix

\section{Initialization and migration workflow}
\label{app:workflow}

This appendix complements the formal specification in \Cref{sec:Persistent BitTorrent Tracker}.
\Cref{fig:workflow} identifies four principals: the \texttt{RepFactory} smart contract, the Reputation Smart Contract, the Tracker (running inside an optional TEE), and BitTorrent peers.
If the tracker operator is trusted, the TEE and attestation step are not needed; all remaining guarantees (portable on-chain reputation, verifiable cryptographic receipts, and transparent migration) hold unchanged.

\begin{figure}[ht]
  \centering
  \includegraphics[width=\linewidth]{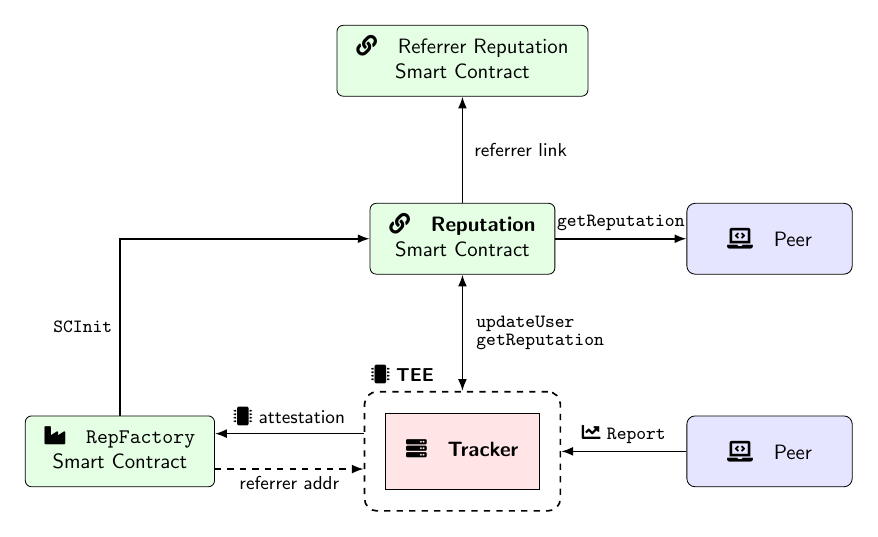}
  \caption{End-to-end workflow of PBTS initialization and operation.
    \textbf{Initialization (left):} The tracker instance, running inside a TEE, submits its TEE attestation to \texttt{RepFactory} together with the address of the predecessor reputation contract (\emph{referrer addr}).
    \texttt{RepFactory} verifies the attestation, then deploys a new Reputation Smart Contract via \texttt{SCInit}, embedding a referrer link to the predecessor so that prior reputation history remains accessible.
    \textbf{Normal operation (centre):} The tracker reads and writes peer reputation on the Reputation Smart Contract via \texttt{getReputation} and \texttt{updateUser}.
    Uploading peers submit aggregated BLS receipts to the tracker via \texttt{Report}; the tracker verifies the aggregate and calls \texttt{updateUser} to credit both uploader and downloaders.
    \textbf{Reputation reads (right):} Any peer, including those operating over the DHT fallback, can call \texttt{getReputation} directly on the smart contract without involving the tracker.
    \textbf{Migration:} When a tracker shuts down, its successor repeats the initialization flow; the referrer link in the new contract allows clients to walk the chain of contracts and reconstruct full reputation history.
  }
  \label{fig:workflow}
\end{figure}

\paragraphNoSkip{Principals}
\texttt{RepFactory} is a publicly deployed, immutable factory contract.
It is the single gatekeeper that controls which tracker instances may deploy reputation contracts, optionally enforcing TEE attestation in the untrusted-operator model.
The Reputation Smart Contract is deployed by \texttt{RepFactory} on behalf of a tracker; it is write-accessible only to the owning tracker instance (via \texttt{updateUser}) and read-accessible to anyone (\texttt{getReputation} is a free view call that incurs no gas cost).
The Tracker processes peer registrations, announce requests, and upload reports; in the untrusted-operator variant it runs inside a TEE so that neither the operator nor the host OS can observe peer IP addresses or tamper with reputation updates.
Peers interact with the tracker for swarm discovery and report submission, and directly with the smart contract for reputation reads.

\paragraphNoSkip{Initialization}
When a new tracker instance starts, it contacts \texttt{RepFactory} with the on-chain address of any predecessor reputation contract it wishes to reference as \emph{referrer addr}.
If the tracker operator is untrusted, the tracker also presents a TEE attestation (a hardware-signed measurement of the tracker binary and configuration). \texttt{RepFactory} verifies this attestation against an approved allowlist before proceeding.
If the tracker operator is trusted, no attestation is required; the tracker authenticates via its on-chain identity (e.g., a signature from a pre-approved address).
In either case, \texttt{RepFactory} then calls \texttt{SCInit} to deploy a fresh Reputation Smart Contract whose referrer pointer is set to the provided address.
The tracker is returned the address of the newly deployed contract and begins accepting peer connections.

\paragraphNoSkip{Normal operation}
Peers register with the tracker by proving ownership of their on-chain public key.
When a peer joins a swarm, the tracker calls \texttt{getReputation} to check that the peer's reputation meets the community threshold before granting access.
After uploading pieces, a peer collects cryptographic receipts from each downloader and submits them in a batch \texttt{Report}.
The tracker verifies the aggregate BLS signature over all receipts and calls \texttt{updateUser} once per epoch boundary to commit the updated upload and download counters to the contract.

\paragraphNoSkip{Migration}
When a tracker becomes unavailable, a successor tracker starts the initialization flow with the old contract's address as \emph{referrer addr}.
Peers resolve the current active contract by reading \texttt{RepFactory}'s registry and, if needed, following the referrer chain.
Because reputation data in old contracts is immutable and publicly readable, no data is lost during migration; peers seamlessly continue accumulating reputation on the new contract while their historical record remains anchored to the predecessor.

\end{document}